\documentclass[aps,prl,reprint,superscriptaddress,longbibliography]{revtex4-2}

\usepackage{amsmath,amssymb,amsthm,mathtools,bm}
\usepackage{graphicx}
\usepackage[svgnames]{xcolor}
\usepackage[colorlinks=true,linkcolor=blue,citecolor=blue,urlcolor=blue]{hyperref}
\usepackage{docmute}

\newtheorem{proposition}{Proposition}
\newtheorem{lemma}{Lemma}
\newtheorem{corollary}{Corollary}

\newcommand{\E}{\mathbb{E}}
\newcommand{\Prob}{\mathbb{P}}
\newcommand{\eps}{\varepsilon}
\newcommand{\rvxi}{\boldsymbol{\xi}}
\newcommand{\rvp}{\boldsymbol{p}}
\newcommand{\rvLambda}{\boldsymbol{\Lambda}}
\newcommand{\rvmu}{\boldsymbol{\mu}}
\newcommand{\guessLambda}{g_\Lambda}
\newcommand{\guessmu}{g_\mu}
\newcommand{\rvM}{\mathbf{M}}
\newcommand{\rvZ}{\mathbf{Z}}

\newcounter{suppsection}
\renewcommand{\thesuppsection}{\Roman{suppsection}}
\makeatletter
\newcommand{\suppsection}[1]{%
  \stepcounter{suppsection}%
  \section{Supplemental Material \thesuppsection: #1}%
  \edef\@currentlabel{\thesuppsection}%
}
\makeatother

\newcommand{\ratepanel}[2]{%
  \begin{minipage}[t]{0.48\textwidth}
  \centering
  \textbf{#1}\\[-0.2em]
  \includegraphics[width=\linewidth]{#2}
  \end{minipage}%
}

\newcommand{\columnratepanel}[2]{%
  \begin{minipage}[t]{0.98\columnwidth}
  \centering
  \textbf{#1}\\[-0.2em]
  \includegraphics[width=\linewidth]{#2}
  \end{minipage}%
}

\begin{document}

\let\combinedbibliography\bibliography
\renewcommand{\bibliography}[1]{}

\title{Quantum Key Distribution Beyond Stationary Channels}

\author{Vaisakh Mannalath}
\email{vmannalath@vqcc.uvigo.es}
\affiliation{Vigo Quantum Communication Center, University of Vigo, Vigo E-36310, Spain}
\affiliation{Escuela de Ingenier{\'i}a de Telecomunicaci{\'o}n, Department of Signal Theory and Communications, University of Vigo, Vigo E-36310, Spain}
\affiliation{AtlanTTic Research Center, University of Vigo, Vigo E-36310, Spain}
\author{V{\'\i}ctor Zapatero}
\affiliation{Vigo Quantum Communication Center, University of Vigo, Vigo E-36310, Spain}
\affiliation{Escuela de Ingenier{\'i}a de Telecomunicaci{\'o}n, Department of Signal Theory and Communications, University of Vigo, Vigo E-36310, Spain}
\affiliation{AtlanTTic Research Center, University of Vigo, Vigo E-36310, Spain}

\author{Kiyoshi Tamaki}
\affiliation{Faculty of Engineering, University of Toyama, Gofuku 3190, Toyama, 930-8555, Japan}

\author{Marcos Curty}
\affiliation{Vigo Quantum Communication Center, University of Vigo, Vigo E-36310, Spain}
\affiliation{Escuela de Ingenier{\'i}a de Telecomunicaci{\'o}n, Department of Signal Theory and Communications, University of Vigo, Vigo E-36310, Spain}
\affiliation{AtlanTTic Research Center, University of Vigo, Vigo E-36310, Spain}

\begin{abstract}
Quantum key distribution (QKD) over non-stationary channels, such as satellite links, is characterized by short, high-loss, and strongly fluctuating transmission windows that produce sparse detection events. In many QKD protocols, these data must be analyzed using non-IID statistical inequalities, yet existing methods either become loose for small sample sizes or heavily rely on fine-tuning, yielding poor estimates when the optical channel is mis-modeled. Using mixture martingale techniques, we introduce tight concentration inequalities that retain sharpness when the channel model is accurate, while remaining robust to model mismatch. In realistic simulations of satellite QKD with fluctuating loss, the resulting bounds can reduce the minimum required number of transmitted signals by more than $70\%$.
\end{abstract}

\maketitle

\section{Introduction}

Quantum key distribution (QKD) enables key expansion secure against computationally unbounded adversaries~\cite{portmann2022security}.
In practical implementations, however, security  is certified from finite data, so the tightness of the statistical bounds used in security proofs directly affects the achievable secret key rate~\cite{lim2021smallblock,mannalath2025sharp}.
Sharper finite-key inference also reduces the number of rounds required to establish a key, and hence the resources needed for deployment.
This issue is especially acute for satellite-based QKD~\cite{liao2017satellite,yin2017satellite_ground_entanglement,chen2021integrated,li2025microsatellite}, one of the clearest near-term routes to long-distance QKD, where each useful pass provides a short, high-loss, non-stationary channel~\cite{sidhu2022finite,islam2024finite_resource,trinh2022statistical,mannalath2026geo}.

In finite-key security proofs against coherent attacks, the number of single-photon detections and phase errors determining the key-length  must be bounded in an adaptive setting~\cite{boileau2005three_state,tamaki2009bennett,curraslorenzo2021tf,mizutani2026passive}.
In this setting, the conditional probability of a later detection event may depend on the record of earlier events. This essentially reduces to solving the following problem. Let $\{\rvxi_i\}_{i=1}^N$ be a sequence of $[0,1]$-valued random variables adapted to a filtration $\{\mathcal F_i\}_{i=0}^N$, where $\mathcal F_{i-1}$ is a $\sigma$-algebra modeling the information available before round $i$. Define $\rvLambda_N=\sum^N_{i=1}\rvxi_i$ and $\rvmu_N=\sum^N_{i=1}\rvp_i$, where $\rvp_i:=\mathbb E[\rvxi_i\vert\mathcal F_{i-1}]$ is the conditional expectation of $\rvxi_i$ given $\mathcal F_{i-1}$. The goal is to quantify the closeness between $\rvLambda_N$ and $\rvmu_N$ via statistical bounds.
More precisely, for a target failure probability $\eps$, we seek one-sided bounds of the form
\begin{equation}
    \Prob\!\left[\rvmu_N>U^\mu(\rvLambda_N)\right]\le\eps,
    \qquad
    \Prob\!\left[\rvLambda_N<L^\Lambda(\rvmu_N)\right]\le\eps,
\end{equation}
for certain functions $U^\mu(\cdot)$ and $L^\Lambda(\cdot)$.
The complementary lower-$\rvmu_N$ and upper-$\rvLambda_N$ bounds are obtained by applying the same construction to the complementary sequence $\{1-\rvxi_i\}_{i=1}^N$.

The standard way to handle this problem is a martingale construction followed by the application of Markov's inequality.
Azuma--Hoeffding bounds~\cite{azuma1967weighted} are a widely used example of this approach, but they are very loose when the number of counts is small~\cite{curraslorenzo2021tf}.
Kato bounds~\cite{kato2020concentration}, on the other hand, reduce this sparse-count penalty by incorporating a guess of $\rvLambda_N$ or $\rvmu_N$ based on prior knowledge.
These latter bounds have enabled tighter finite-key analyses in several recent QKD security proofs~\cite{curraslorenzo2021tf,curraslorenzo2021randomsampling,mizutani2026passive,zapatero2024passive,curraslorenzo2025securityframework,navarrete2026numerical}.
Their sharpness, however, is local to the chosen guess, and they are not very robust to incorrect guesses.

This guess sensitivity is a severe limitation for non-stationary channels like satellite QKD links, whose optical loss varies from pass to pass rather than remaining fixed~\cite{liorni2019satellite,scarfe2025adaptive}.
Indeed, LEO satellite experiments and related studies report link-budget shifts from satellite elevation, diffraction, atmospheric propagation and tracking, with magnitudes up to tens of decibels~\cite{lu2022micius,ecker2021strategies,takenaka2017satellite}.
That is, the transmittance of a satellite quantum channel is naturally modeled as a random variable~\cite{trinh2022statistical}, and a guess selected according to the expected channel can be displaced by several decibels in an actual implementation.
As a consequence, satellite-based QKD needs sharp non-IID bounds robust in non-stationary scenarios.

In this Article, we solve this problem by introducing guess-robust concentration bounds valid for non-IID data.
Remarkably, they retain the sparse-count sharpness of Kato bounds without committing to a single channel behavior.
In a nutshell, the key idea goes as follows.
We first construct a martingale, in the same spirit as Kato bounds, dependent on a parameter $\lambda$ that can be finely tuned for the expected channel behavior.
Importantly, the resulting bound is much less sensitive than Kato's to channel deviations.
We then implement a mixing technique that averages over $\lambda$ with a suitable weight function~\cite{robbins1969confidence,robbins1970statistical,lai1976confidence,kaufmann2021mixture}, trading local sharpness for robustness across channel mismatch.
As a result, the closed-form confidence bounds significantly outperform Kato's in non-stationary channels, and can be inserted directly into finite-key QKD security analyses. To demonstrate the tightness of our approach, we benchmark the results against the binomial tail, which provides a natural limit on how sharp a general non-IID bound can be. In a decoy-state BB84 simulation over a satellite link, our method remains robust under significant channel mismatch and substantially reduces the minimum number of rounds for positive-key extraction.

\section{Results}

\subsection{Martingale concentration bounds}

We use the notation in the Introduction, with bold symbols denoting random variables.
First, we construct a martingale with a tunable parameter $\lambda\geq0$,
\begin{equation}
    \rvM_n(\lambda):=
    \prod_{i=1}^n
    \frac{e^{-\lambda \rvxi_i}}
    {\E[e^{-\lambda \rvxi_i}\mid\mathcal F_{i-1}]} .
    \label{eq:fixed_lambda_martingale_main}
\end{equation}

To relate the denominator product to $\rvmu_N$, we first use $e^{-\lambda z}\le 1-(1-e^{-\lambda})z$ for $z\in[0,1]$, which gives $\E[e^{-\lambda\rvxi_i}\mid\mathcal F_{i-1}]\le 1-(1-e^{-\lambda})\rvp_i$.
We then apply Jensen's inequality~\cite{jensen1906fonctions} to the concave function $t\mapsto\ln[1-(1-e^{-\lambda})t]$, obtaining $\sum_{i=1}^N\ln[1-(1-e^{-\lambda})\rvp_i]\le N\ln[1-(1-e^{-\lambda})\rvmu_N/N]$.
It follows that

\begin{equation}
    \begin{aligned}
        \rvM_N(\lambda)
         & =
        e^{-\lambda \rvLambda_N}
        \prod_{i=1}^N
        \E[e^{-\lambda\rvxi_i}\mid\mathcal F_{i-1}]^{-1}
        \\
         & \ge
        e^{-\lambda \rvLambda_N}
        \prod_{i=1}^N
        \left[1-(1-e^{-\lambda})\rvp_i\right]^{-1}
        \\
         & \ge
        e^{-\lambda \rvLambda_N}
        \left[1-(1-e^{-\lambda})\frac{\rvmu_N}{N}\right]^{-N}
        =:
        K_\lambda(\rvLambda_N,\rvmu_N).
    \end{aligned}
\end{equation}
Consequently, $\{\rvM_N(\lambda)\ge1/\eps\}\supseteq\{K_\lambda(\rvLambda_N,\rvmu_N)\ge1/\eps\}$, and since $\rvM_N(\lambda)$ is nonnegative with mean one, Markov's inequality~\cite{steinshakarchi2005realanalysis} gives
\begin{equation}
    \Prob\!\left[K_\lambda(\rvLambda_N,\rvmu_N)\ge1/\eps\right]\le\eps.
\end{equation}
For a fixed $\lambda$, the tail inequality above can be inverted to obtain statistical bounds relating $\rvLambda_N$ and $\rvmu_N$, which we refer to as fixed-$\lambda$ martingale bounds.
These analytical bounds are given in Supplemental Material IV\nocite{apostol1974mathematicalanalysis,kechris1995classical,klenke2008probability,mastin_jaillet2013logquadratic,tonelli1909integrazione}, where a guess on $\rvLambda_N$ or $\rvmu_N$ based on prior knowledge is converted into a fixed value of $\lambda$.

To finely tune $\lambda$ for an expected channel behavior, consider a realization $(Nx,Ny)$ of $(\rvLambda_N,\rvmu_N)$, with $0<x\le y<1$.
For that realization, the choice $\lambda^*(x,y)=\ln\!\left[y(1-x)/x(1-y)\right]$ maximizes $K_\lambda(Nx,Ny)$, giving
\begin{equation}
    K_{\lambda^*(x,y)}(Nx,Ny)=e^{N D(x\|y)}.
    \label{eq:fixed_lambda_kl_match_main}
\end{equation}
Here, $D(x\|y):=x\ln(x/y)+(1-x)\ln[(1-x)/(1-y)]$ is the Bernoulli Kullback--Leibler (KL) divergence.
Indeed, the martingale in Eq.~\eqref{eq:fixed_lambda_martingale_main} has been chosen specifically to recover this KL exponent at the optimizer, which also
occurs in the Chernoff bound~\cite{chernoff1952measure} for a binomial distribution with mean fraction $y$ and observed fraction $x$.
Thus, at the optimal point, the fixed-$\lambda$ martingale recovers the IID Chernoff exponent while remaining valid in the non-IID setting.

In what follows, we remove the dependence on $\lambda$ by averaging this parameter according to a prior distribution $\pi$ of admissible values. This results in a so-called mixture martingale~\cite{robbins1969confidence,robbins1970statistical,lai1976confidence,kaufmann2021mixture}, of which the fixed-$\lambda$ martingale represents a particular case,  $\pi=\delta_{\lambda}$. Taking instead $\pi$ to be the uniform distribution on $[A,B]$, we get
\begin{equation}
    \widetilde{\rvM}_n:=
    \frac{1}{B-A}\int_A^B \rvM_n(\lambda)\,d\lambda .
\end{equation}
One can readily show that this is again a nonnegative martingale with mean one.
Repeating the inequality scaling argument above, we have
\begin{equation}
    \widetilde{\rvM}_N
    \ge
    \frac{1}{B-A}\int_A^B K_\lambda(\rvLambda_N,\rvmu_N)\,d\lambda
    =:
    K_{A,B}(\rvLambda_N,\rvmu_N),
    \label{eq:mix-kernel-main}
\end{equation}
and after applying Markov's inequality, we obtain
\begin{equation}
    \Prob\!\left[K_{A,B}(\rvLambda_N,\rvmu_N)\ge 1/\eps\right]\le\eps.
\end{equation}
The interval-mixture bounds are then calculated by inverting $K_{A,B}$ in the variable to be bounded, leading to the following proposition.

\begin{proposition}\label{thm:main-prl}
    \begingroup
    \small
    \setlength{\jot}{2pt}
    Set $\lambda\ge0$, $0\le A<B<\infty$, and $\eps\in(0,1)$.
    Let $\{\rvxi_i\}_{i=1}^N$ be a stochastic process adapted to the filtration $\{\mathcal F_i\}_{i=0}^N$, with $\rvxi_i\in[0,1]$ and $\rvp_i=\E[\rvxi_i\vert\mathcal F_{i-1}]$.
    We write $\rvLambda_N=\sum_{i=1}^N\rvxi_i$ and $\rvmu_N=\sum_{i=1}^N\rvp_i$, and denote by $K_I$ either the fixed-$\lambda$ function $K_\lambda$ or the mixture-martingale function $K_{A,B}$.
    For $s,u\in[0,N]$, define
    \begin{equation}
    \begin{aligned}
        U_{I}^{\mu}(s)
         & :=
        \inf\!\Bigl[\{v\in[0,N]:K_I(s,v)>1/\eps\}\cup\{N\}\Bigr],
        \\
        L_{I}^{\Lambda}(u)
         & :=
        \sup\!\Bigl[\{r\in[0,N]:K_I(r,u)>1/\eps\}\cup\{0\}\Bigr].
    \end{aligned}
    \label{eq:main_inversion_maps}
    \end{equation}
    Then,
    \begin{equation}
        \Prob\!\left[\rvmu_N>U_{I}^{\mu}(\rvLambda_N)\right]\le \eps,\qquad \Prob\!\left[\rvLambda_N<L_{I}^{\Lambda}(\rvmu_N)\right]\le \eps.
        \label{eq:mu-upper-main}
    \end{equation}
    The complementary lower-$\rvmu_N$ and upper-$\rvLambda_N$ bounds are obtained by applying the same inversions to the complementary variables $\bar{\rvxi}_i:=1-\rvxi_i$, $\bar{\rvp}_i:=1-\rvp_i$, $\bar{\rvLambda}_N:=N-\rvLambda_N$, and $\bar{\rvmu}_N:=N-\rvmu_N$.
    \endgroup
\end{proposition}

A proof of Proposition~\ref{thm:main-prl} is given in Supplemental Material I.

Now, we relate $K_{A,B}$ to the Bernoulli KL divergence.
For this purpose, suppose that the optimizer lies inside the mixture interval, i.e., $\lambda^*(x,y)\in[A,B]$.
Then, it is shown in Supplemental Material III that
\begin{equation}
    K_{A,B}(Nx,Ny)
    \ge
    e^{N D(x\|y)}
    \left[1+\frac{N(B-A)^2}{2\pi}\right]^{-1/2}.
    \label{eq:mix-kl-factor-main}
\end{equation}
Equation~\eqref{eq:mix-kl-factor-main} shows that, independently of the realization, the departure from the optimal fixed-$\lambda$ performance induced by the mixing is only affected by a small penalty dependent on the interval width.

In Supplemental Material III, we combine Eq.~\eqref{eq:mix-kl-factor-main} with rational KL approximations~\cite{mannalath2025sharp,topsoe2007some} to obtain analytical mixture-martingale bounds. Precisely, while Ref.~\cite{mannalath2025sharp} derives bounds for sums of independent Bernoulli trials, here, the same relaxation yields a martingale counterpart valid in the non-IID setting, with the same algebraic form up to a rescaling of the failure probability induced by the mixing.

A natural question is how much tighter a general non-IID bound can be compared to our results. In this regard, as discussed in Supplemental Material V, the binomial-tail inversions underlying the Clopper-Pearson bounds~\cite{clopper1934confidence} are uniformly tighter than any one-sided bound valid for general non-IID processes (whether the latter follows the martingale-Markov route or not).
Therefore, we can use the former for benchmarking purposes.
To make this benchmark concrete, we define the lower binomial tail $B_N(s,u):=\Prob[\operatorname{Bin}(N,u/N)\le s]$.
The Clopper--Pearson benchmark is obtained from Eq.~\eqref{eq:main_inversion_maps} by replacing the inversion condition $K_I(s,u)>1/\eps$ with $B_N(s,u)<\eps$.
This benchmark also gives the right scale for comparison with Eq.~\eqref{eq:mix-kl-factor-main}.
Writing $s=\lfloor Nx\rfloor$ and $u=Ny$ with fixed $0<x<y<1$, the sharp binomial tail asymptotics give~\cite{bahadurrao1960deviations,arratia1989tutorial,zhu2022nearly}

\begin{equation}
\label{eq:binomial-asymptotics}
    B_N(s,u)
    =
    e^{-N D(x\|y)+O(1)}N^{-1/2}.
\end{equation}
Comparing Eqs.~\eqref{eq:mix-kl-factor-main} and \eqref{eq:binomial-asymptotics}, we find that the mixture martingale has the same leading $N D(x\|y)$ exponent whenever $\lambda^*(x,y)\in[A,B]$, with only an $O(1)$ difference inside the exponent.

\begin{figure}[!htbp]
    \centering
    \includegraphics[width=0.95\columnwidth]{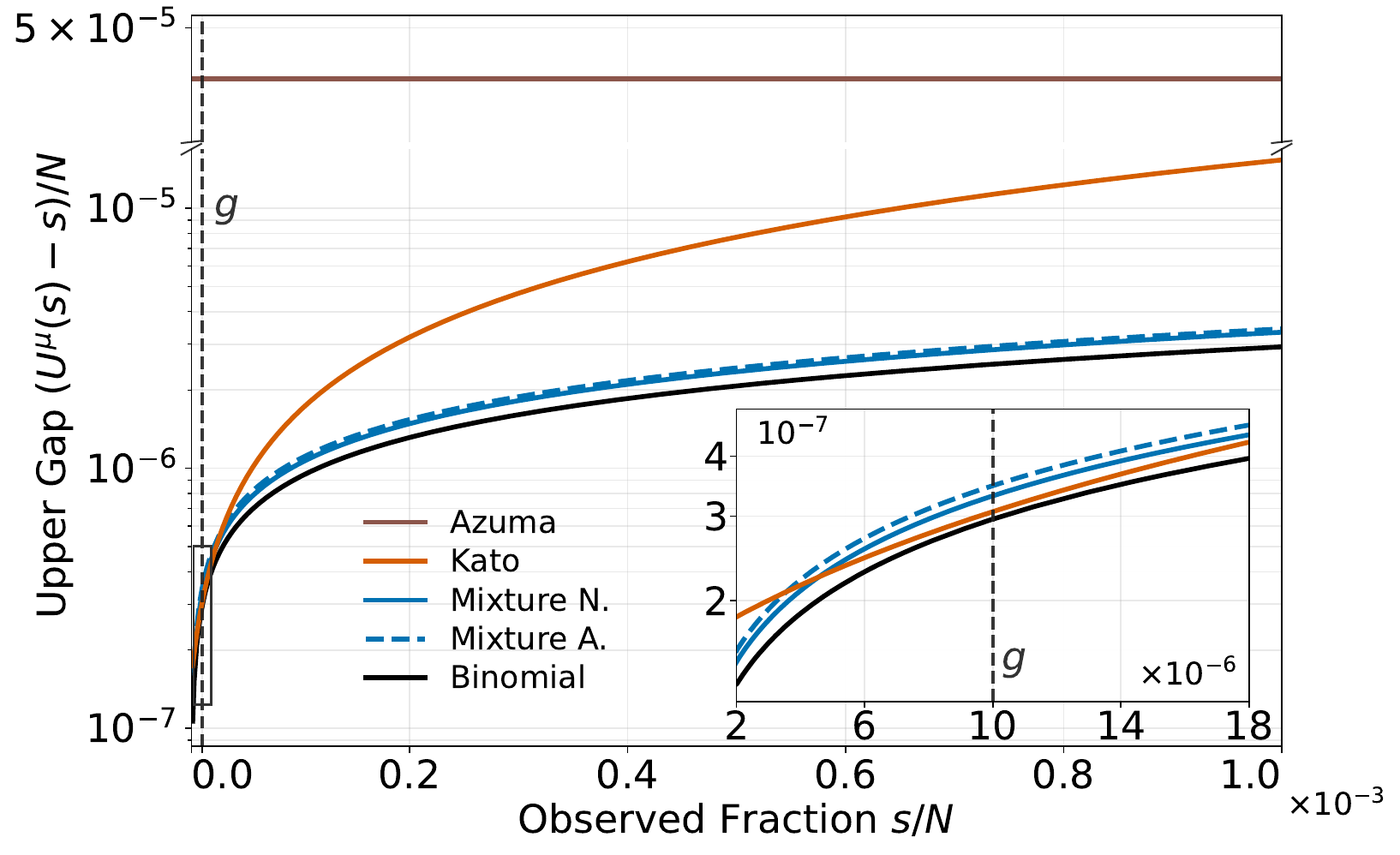}
    \caption{Comparison of upper bounds on $\mu$ for $N=10^{10}$ and $\eps=10^{-20}$, plotted as the normalized gap $[U^\mu(s)-s]/N$ versus $x=s/N$.
    The dashed vertical line marks the normalized guess $g_\Lambda/N=10^{-5}$ and the inset magnifies the neighborhood of $g_\Lambda/N$.
    Blue curves use the numerical mixture-martingale bound of Proposition~\ref{thm:main-prl} with $[A,B]=[0,10]$ (solid) and its corresponding analytical relaxation (dashed).
    The remaining curves are the Clopper--Pearson binomial benchmark (black), the Kato bound~\cite{kato2020concentration} (orange), and the Azuma--Hoeffding bound~\cite{azuma1967weighted} (brown).}
    \label{fig:mu-gap}
\end{figure}
This closeness is also visible in Fig.~\ref{fig:mu-gap}, which compares the mixture-martingale upper bound of Proposition~\ref{thm:main-prl} with $[A,B]=[0,10]$, the Kato upper bound~\cite{kato2020concentration} with normalized guess $g_\Lambda/N=10^{-5}$, the Azuma--Hoeffding upper bound~\cite{azuma1967weighted}, and the Clopper--Pearson binomial benchmark for $\rvmu_N$ with $N=10^{10}$ and $\eps=10^{-20}$.
As expected, the interval-mixture bound remains very close to the binomial benchmark across the plotted range.
The Kato bound is sharp near its guess fraction, $g_\Lambda/N$, but deteriorates as the observed fraction $x$ moves away from that calibration. Azuma--Hoeffding, which is independent of $x$, remains visibly looser throughout this sparse-count window. The formulas for the Kato and Azuma--Hoeffding bounds are given in Supplemental Material VI and VII, respectively.
At this scale, the fixed-$\lambda$ bound is barely distinguishable from Kato's.
See Fig.~S1 in the Supplemental Material for a separate comparison of Kato and the fixed-$\lambda$ bound using a different guess value and $x$-axis range.

\subsection{Simulations}
Let us now compare the performance of the bounds for a decoy-state BB84 protocol over a satellite link.
In each protocol round, Alice, at the satellite, sends Bob a phase-randomized weak coherent pulse chosen from three intensities and two bases.
Bob, at the optical ground station, measures the incoming pulse with an asymmetric passive BB84 receiver~\cite{mizutani2026passive}, splitting it between $Z$- and $X$-basis detection arms.
He records $Z$-arm clicks, $X$-arm clicks, and cross-clicks, with same-arm double clicks assigned uniformly at random to single-click events.
After sifting, the parties apply the non-IID bounds to the intensity-resolved $Z$-basis counts, $X$-basis error counts, and cross-click counts to estimate the single-photon quantities that determine the secret key length~\cite{koashi2009simple}. 

The channel model is characterized by the total system transmittance $\eta_{\rm sys}$, the per-detector dark-count probability $d$ and the misalignment contribution $\delta_{\mathrm{mis}}$.
We use an overall secrecy parameter $\eps_{\mathrm{sec}}\approx10^{-10}$, with each concentration failure probability set to $\eps=10^{-20}/144\approx7\times10^{-23}$.
The exact protocol steps and simulation details are provided in Supplemental Material VIII.

Figures~\ref{fig:adaptive-key} and~\ref{fig:geo-adaptive-key} plot the finite-key rate versus the number of protocol rounds for the different concentration bounds.
In particular, in Fig.~\ref{fig:adaptive-key} we set the expected channel loss to $L_{\rm exp}=30\,\mathrm{dB}$ and the observed channel loss to $L_{\rm obs}=35$ (Fig.~\ref{fig:adaptive-key}a) and $40\,\mathrm{dB}$ (Fig.~\ref{fig:adaptive-key}b), and take $d=10^{-7}$ and $\delta_{\mathrm{mis}}=0.03$, corresponding to typical channel parameters for a LEO satellite link~\cite{liao2017satellite,sidhu2022finite,avesani2021daylight}.
In Fig.~\ref{fig:geo-adaptive-key}, we set $L_{\rm exp}=50\,\mathrm{dB}$ and $L_{\rm obs}=55$ (Fig.~\ref{fig:geo-adaptive-key}a) and $60\,\mathrm{dB}$ (Fig.~\ref{fig:geo-adaptive-key}b), and take $d=10^{-9}$ and $\delta_{\mathrm{mis}}=0.01$, emulating a GEO-QKD link with higher losses and detector-limited noise~\cite{mannalath2026geo}.
In both plots, the decoy intensities and their probabilities, the basis probabilities, and the passive beam-splitter ratio are optimized to maximize the secret key rate using the \emph{expected-loss} channel model.
These parameters are then used to evaluate the finite-key rate with the \emph{observed-loss} channel model.
For the two guess-dependent bound families, Kato~\cite{kato2020concentration} and the fixed-$\lambda$ bounds, the guess parameters are constructed from the expected-loss channel model.
For the mixture-martingale bounds, we use the same integration range $[A,B]=[0,10]$ throughout the plots, 
so that the relevant optimizer $\lambda^*(x,y)$ for the analytical guarantee remains inside the interval in all the simulations.
We note that the range $[A,B]$ could be tuned separately for each concentration bound to improve performance.
In practice, channel variation can guide this choice: for each bound, one could evaluate $\lambda^*(x,y)$ at the channel extremes and set $A$ and $B$ to the smallest and largest values obtained.
However, using a common range as we do here keeps the comparison simple and demonstrates robustness without fine-tuning.

As a result, we find that the mixture-martingale bounds reduce the threshold number of rounds by about $13\%$ and $5\%$ in the LEO and GEO benchmarks, respectively, for a $5\,\mathrm{dB}$ loss mismatch.
For a $10\,\mathrm{dB}$ loss mismatch, the corresponding reductions are about $71\%$ and $44\%$, respectively.
The mixture curves track the binomial bounds closely for higher mismatches, while the fixed-$\lambda$ bounds give a smaller but visible improvement over Kato.

In Supplemental Material IX we show additional key-rate plots in which the protocol parameters are optimized for the observed ---rather than the expected--- channel loss. In doing so, we isolate the raw impact of the concentration bounds, unaffected by any incidental mismatch from their optimal parameters (of course, the guesses for the bounds are still calculated assuming the expected channel model).
Under loss mismatch, the mixture-martingale bounds show the same behavior: they improve over Kato and remain closer to the binomial benchmark.
Additionally, we plot the case where there is no loss mismatch.
In this scenario, since Kato and the fixed-$\lambda$ martingale bounds are correctly calibrated, they outperform the mixture-martingale bound.
Notably though, by taking a union bound of the fixed-$\lambda$ and mixture-martingale bounds one can still improve over Kato in all cases.

\begin{figure}[!t]
    \centering
    \columnratepanel{(a)}{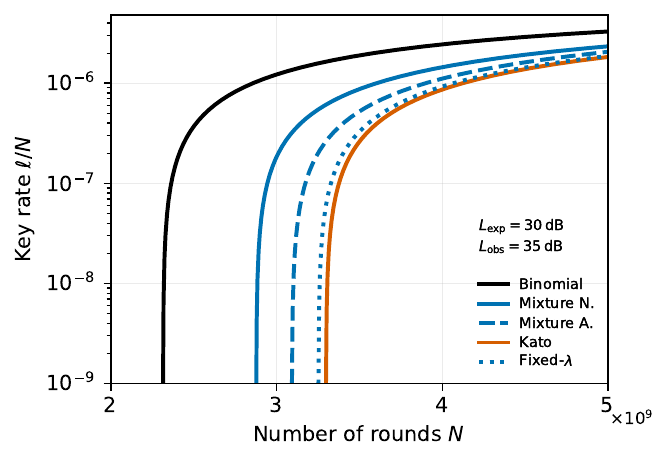}\\[0.6em]
    \columnratepanel{(b)}{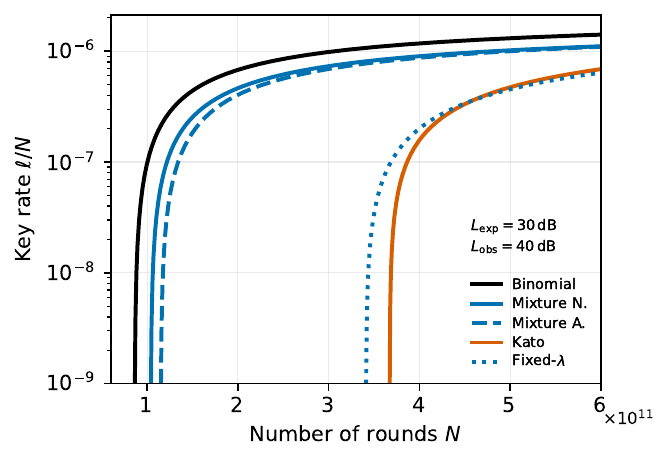}
    \caption{Finite secret key rate versus number of rounds for a decoy-state BB84 protocol with an asymmetric passive receiver over a LEO-satellite link with channel parameters $d=10^{-7}$ and $\delta_{\mathrm{mis}}=0.03$.
    The expected loss is $L_{\rm exp}=30\,\mathrm{dB}$, while panels (a) and (b) use $L_{\rm obs}=35$ and $40\,\mathrm{dB}$, respectively.
    Blue curves use the numerical mixture-martingale bounds of Proposition~\ref{thm:main-prl} with $[A,B]=[0,10]$ (solid), their corresponding analytical relaxations (dashed), and the fixed-$\lambda$ martingale bound of Proposition~\ref{thm:main-prl} (dotted).
    The remaining curves use the Clopper--Pearson binomial benchmark (black) and Kato's bound~\cite{kato2020concentration} (orange).
    Azuma--Hoeffding~\cite{azuma1967weighted} curves are omitted since their positive-key thresholds are orders of magnitude larger than Kato's.}
    \label{fig:adaptive-key}
\end{figure}

\begin{figure}[!t]
    \centering
    \columnratepanel{(a)}{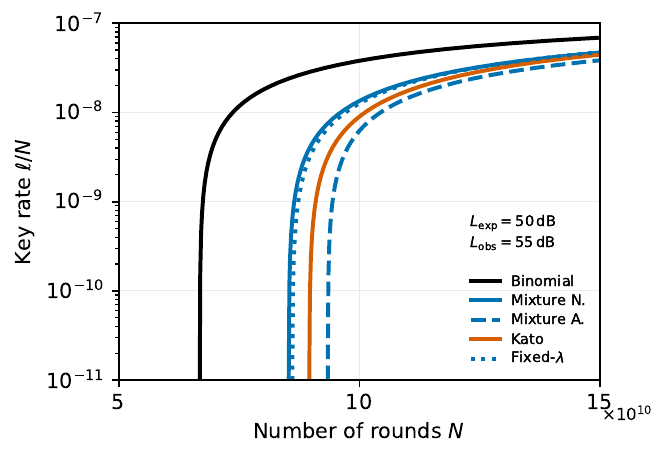}\\[0.6em]
    \columnratepanel{(b)}{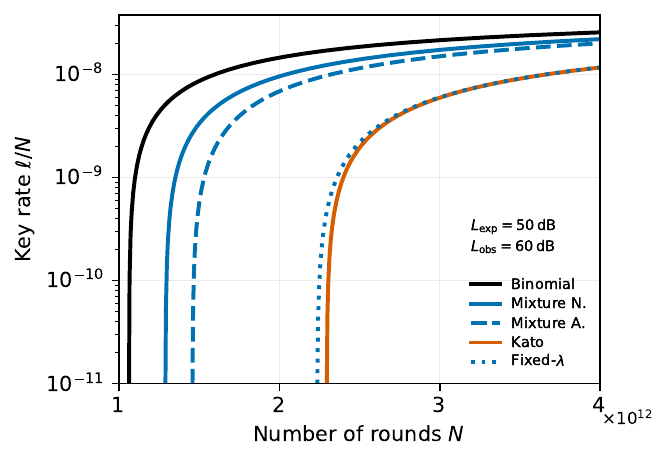}
    \caption{Finite secret key rate versus number of rounds for an asymmetric passive decoy-state BB84 protocol over a GEO-satellite link with channel parameters $d=10^{-9}$ and $\delta_{\mathrm{mis}}=0.01$.
    The expected loss is $L_{\rm exp}=50\,\mathrm{dB}$, while panels (a) and (b) use $L_{\rm obs}=55$ and $60\,\mathrm{dB}$, respectively.
    All color and line-style conventions match Fig.~\ref{fig:adaptive-key}.}
    \label{fig:geo-adaptive-key}
\end{figure}

\section{Discussion}
Finite statistics are a practical bottleneck for QKD.
This is especially acute in satellite-based QKD, where each overpass gives only a limited acquisition window, and the realized loss can deviate substantially from the design value.
In this regime, sharp statistics can save emitted pulses, acquisition time and link resources, while a bound calibrated to the wrong channel can turn usable data into no secret key rate.

In this work, we have developed mixture martingale bounds for adaptive processes, which provide robustness against non-stationary channels and are straightforward to insert into existing finite-key analyses.
Compared with Kato bounds~\cite{kato2020concentration}, the construction retains sparse-count sharpness near the expected channel behavior, and significantly reduces the penalty when the channel conditions change, enabling key extraction with substantially lower data acquisition.
The comparison with the binomial benchmark highlights the tightness of the developed bounds.
The interval width, in turn, constitutes a practical design knob: a wider interval protects against larger channel variation, whereas a narrower one gives tighter bounds for a well-characterized link.
This makes statistical robustness part of the same design space as loss, background, and acquisition time.

These results are also relevant to other quantum communication protocols involving non-i.i.d. processes with sparse data and uncertain operating conditions.

\section*{Data availability}

The data and code supporting the findings of this study are available from the authors upon reasonable request.

\section*{Acknowledgements}
    We acknowledge valuable discussions with Go Kato. This work was supported by the European Union’s Horizon Europe Framework Programme under the Marie Sklodowska-Curie Grant No. 101072637 (Project QSI), the Galician Regional Government (consolidation of research units: atlanTTic), the Spanish Ministry of Science, Innovation and Universities (MICIU), the Fondo Europeo de Desarrollo Regional (FEDER) through the grant No. PID2024-162270OB-I00, the “Hub Nacional de Excelencia en Comunicaciones Cuanticas” funded by the Spanish Ministry for Digital Transformation and the Public Service and the European Union NextGenerationEU, the European Union’s Horizon Europe Framework Programme under the project ``Quantum Secure Networks Partnership" (QSNP, grant agreement No 101114043), the European Union under the Project IberianQCI (grant 101249593), and the “Programa de Cooperación Interreg VI-A España–Portugal” (POCTEP) 2021–2027 through the project QUANTUM IBER\_IA. K.T. acknowledges support from JSPS KAKENHI Grant Numbers 23K25793 and 23H01096.
An OpenAI GPT model assisted with the literature survey and parts of the technical proofs; the authors take full responsibility for the final manuscript.

\section*{Author contributions}

K.T. and V.M. conceptualized the study.
M.C. established the overall structure and methodological framework of the study.
V.M. and V.Z. performed the formal analysis.
V.M. developed the simulations and wrote the original draft.
M.C., V.Z., and K.T. supervised the work.
M.C. and V.Z. made substantial revisions to the manuscript.
All authors reviewed and edited the manuscript.

\section*{Competing interests}

The authors declare no competing interests.

\bibliography{refs}

\clearpage

\begin{widetext}

\begin{center}
{\large\bfseries Supplemental Material for ``Quantum Key Distribution Beyond Stationary Channels''}\\[0.6em]
Vaisakh Mannalath,$^{1,2,3}$ V{\'\i}ctor Zapatero,$^{1,2,3}$ Kiyoshi Tamaki,$^{4}$ and Marcos Curty$^{1,2,3}$\\[0.5em]
{\small
$^{1}$Vigo Quantum Communication Center, University of Vigo, Vigo E-36310, Spain\\
$^{2}$Escuela de Ingenier{\'i}a de Telecomunicaci{\'o}n, Department of Signal\\ Theory and Communications, University of Vigo, Vigo E-36310, Spain\\
$^{3}$AtlanTTic Research Center, University of Vigo, Vigo E-36310, Spain\\
$^{4}$Faculty of Engineering, University of Toyama, Gofuku 3190, Toyama, 930-8555, Japan
}
\end{center}
\end{widetext}
\appendix
\setcounter{section}{0}
\setcounter{subsection}{0}
\setcounter{equation}{0}
\setcounter{figure}{0}
\setcounter{table}{0}
\setcounter{theorem}{0}
\setcounter{proposition}{0}
\setcounter{lemma}{0}
\setcounter{corollary}{0}
\setcounter{suppsection}{0}
\renewcommand{\theequation}{S\arabic{equation}}
\renewcommand{\thefigure}{S\arabic{figure}}
\renewcommand{\thetable}{S\arabic{table}}
\renewcommand{\thetheorem}{S\arabic{theorem}}
\renewcommand{\theproposition}{S\arabic{proposition}}
\renewcommand{\thelemma}{S\arabic{lemma}}
\renewcommand{\thecorollary}{S\arabic{corollary}}
\renewcommand{\theHequation}{supp.\arabic{equation}}
\renewcommand{\theHfigure}{supp.\arabic{figure}}
\renewcommand{\theHtable}{supp.\arabic{table}}
\renewcommand{\theHtheorem}{supp.\arabic{theorem}}
\renewcommand{\theHproposition}{supp.\arabic{proposition}}
\renewcommand{\theHlemma}{supp.\arabic{lemma}}
\renewcommand{\theHcorollary}{supp.\arabic{corollary}}

\suppsection{Mixture-martingale bounds}\label{app:mixture_concentration_bounds}

This section derives the mixture-martingale bounds from Proposition~1 of the Article.

Let $(\Omega,\mathcal F,\mathbb P)$ be a probability space with filtration $(\mathcal F_n)_{n=0}^N$.
For $i=1,\dots,N$, take $\rvxi_i:\Omega\to[0,1]$ to be $\mathcal F_i$-measurable, and define
\begin{equation}
    \rvp_i:=\mathbb E[\rvxi_i\mid \mathcal F_{i-1}],
    \qquad
    \rvLambda_n:=\sum_{i=1}^n \rvxi_i,
    \qquad
    \rvmu_n:=\sum_{i=1}^n \rvp_i.
\end{equation}
Fix $\varepsilon\in(0,1)$ and a probability measure $\pi$ on $[0,\infty)$.

\begin{proposition}\label{thm:mix_joint_tail}
    Define
    \begin{equation}\label{eq:mix_gpi_def}
        \begin{aligned}
            K_\lambda(s,u)
             &:=
            e^{-\lambda s}
            \left[1-(1-e^{-\lambda})\frac{u}{N}\right]^{-N}
            \qquad \lambda\ge0,
            \\
            K_\pi(s,u)
             &:=
            \int_0^\infty K_\lambda(s,u)\,\pi(d\lambda),
            \qquad s,u\in[0,N].
        \end{aligned}
    \end{equation}
    Then, we have that
    \begin{equation}\label{eq:mix_joint_tail}
        \mathbb P\!\left[K_\pi(\rvLambda_N,\rvmu_N)\ge 1/\varepsilon\right]\le \varepsilon.
    \end{equation}
\end{proposition}

\begin{proof}
    Fix $\lambda\ge 0$ and define
    \begin{equation}\label{eq:mix_martingale_def}
        \rvM_n(\lambda)
        :=
        \prod_{i=1}^n
        \frac{e^{-\lambda \rvxi_i}}
        {\mathbb E[e^{-\lambda \rvxi_i}\mid \mathcal F_{i-1}]},
        \qquad n=0,1,\dots,N,
    \end{equation}
    with the empty product interpreted as $1$.
    For $x\in[0,1]$, convexity of $x\mapsto e^{-\lambda x}$ gives
    \begin{equation}
        \label{eq:mix_binary_exp}
        0< e^{-\lambda x}\le 1-(1-e^{-\lambda})x.
    \end{equation}
    Applying \eqref{eq:mix_binary_exp} to the random variable $\rvxi_i$ and then taking conditional expectation yields
    \begin{equation}
        \label{eq:mix_cond_mgf_identity}
        0<\mathbb E\!\left[e^{-\lambda \rvxi_i}\mid \mathcal F_{i-1}\right]
        \le
        1-(1-e^{-\lambda})\rvp_i.
    \end{equation}
    For $n\ge 1$, the definition \eqref{eq:mix_martingale_def} gives
    \begin{equation}\label{eq:mix_martingale_recursion}
        \rvM_n(\lambda)
        =
        \rvM_{n-1}(\lambda)
        \frac{e^{-\lambda \rvxi_n}}{\mathbb E[e^{-\lambda \rvxi_n}\mid \mathcal F_{n-1}]}.
    \end{equation}
    Taking conditional expectation given $\mathcal F_{n-1}$, and using the fact that $\rvM_{n-1}(\lambda)$ is $\mathcal F_{n-1}$-measurable, we obtain
    \begin{equation}\label{eq:mix_martingale_ce}
        \mathbb E\!\left[\rvM_n(\lambda)\mid \mathcal F_{n-1}\right]
        =
        \rvM_{n-1}(\lambda).
    \end{equation}
Hence $\{\rvM_n(\lambda)\}_{n=0}^N$ is a nonnegative martingale, and repeated use of \eqref{eq:mix_martingale_ce} gives    \begin{equation}
        \label{eq:mix_fixed_mean_one}
        \begin{aligned}
            \mathbb E[\rvM_N(\lambda)]
             & =\mathbb E\!\left[\mathbb E\!\left[\rvM_N(\lambda)\mid\mathcal F_{N-1}\right]\right] \\
             & =\mathbb E[\rvM_{N-1}(\lambda)]
            =\cdots
            =\mathbb E[\rvM_0(\lambda)]
            =1.
        \end{aligned}
    \end{equation}
    Now define the terminal mixture statistic, following the standard mixture-martingale construction~\cite{kaufmann2021mixture}, by
    \begin{equation}
        \widetilde{\rvM}_N:=\int_0^\infty \rvM_N(\lambda)\,\pi(d\lambda).
    \end{equation}
    Since $\rvM_N(\lambda)$ is nonnegative, Tonelli's theorem~\cite{tonelli1909integrazione} gives
    \begin{equation}\label{eq:mix_tonelli_mean}
        \begin{aligned}
            \mathbb E[\widetilde{\rvM}_N]
             & =
            \int_0^\infty \mathbb E\!\left[\rvM_N(\lambda)\right]\pi(d\lambda)
            \\
             & =
            \int_0^\infty 1\,\pi(d\lambda)
            =1.
        \end{aligned}
    \end{equation}
    Applying Markov's inequality to $\widetilde{\rvM}_N$, we obtain
    \begin{equation}\label{eq:mix_markov}
        \mathbb P\!\left[\widetilde{\rvM}_N\ge 1/\varepsilon\right]
        \le
        \varepsilon\mathbb E[\widetilde{\rvM}_N]
        =
        \varepsilon.
    \end{equation}
    Next define, for fixed $\lambda$,
    \begin{equation}\label{eq:f_second_derivative}
        f_\lambda(x):=\ln\!\bigl[1-(1-e^{-\lambda})x\bigr],
        \qquad x\in[0,1].
    \end{equation}
    We have that
    \begin{equation}\label{eq:jensen_generic}
        f_\lambda''(x)=-\frac{(1-e^{-\lambda})^2}{\bigl[1-(1-e^{-\lambda})x\bigr]^2}\le 0,
    \end{equation}
    so $f_\lambda (x)$ is concave on $x\in[0,1]$.
    Therefore, by Jensen's inequality~\cite{jensen1906fonctions}, for every $x_1,\dots,x_N\in[0,1]$,
    \begin{equation}
        \frac{1}{N}\sum_{i=1}^N f_\lambda(x_i)\le f_\lambda\!\left(\frac{1}{N}\sum_{i=1}^N x_i\right).
    \end{equation}
    For each $\omega\in\Omega$, apply this inequality to the numbers $\rvp_i(\omega)\in[0,1]$.
    This yields
    \begin{equation}\label{eq:mix_jensen}
        \sum_{i=1}^N \ln\!\bigl[1-(1-e^{-\lambda})\rvp_i\bigr]
        \le
        N\ln\!\left[1-(1-e^{-\lambda})\frac{\rvmu_N}{N}\right].
    \end{equation}
    Exponentiating \eqref{eq:mix_jensen}, taking reciprocals, and multiplying by $e^{-\lambda \rvLambda_N}$  give
    \begin{equation}
        \begin{aligned}
            \prod_{i=1}^N
            \frac{e^{-\lambda \rvxi_i}}
            {1-(1-e^{-\lambda})\rvp_i}
             & \ge
            e^{-\lambda \rvLambda_N} 
            \left[
                1-(1-e^{-\lambda})\frac{\rvmu_N}{N}
            \right]^{-N}.
        \end{aligned}
    \end{equation}
    By \eqref{eq:mix_martingale_def} and \eqref{eq:mix_cond_mgf_identity}, we have that
    \begin{equation}\label{eq:mix_pointwise_bound}
        \rvM_N(\lambda)
        \ge
        \prod_{i=1}^N\frac{e^{-\lambda \rvxi_i}}{1-(1-e^{-\lambda})\rvp_i}.
    \end{equation}
    Combining the last two equations and integrating over $\pi(d\lambda)$ gives
    \begin{equation}
        \widetilde{\rvM}_N
        \ge
        K_\pi(\rvLambda_N,\rvmu_N).
    \end{equation}
    Therefore, it is guaranteed that
    \begin{equation}
        \left\{\widetilde{\rvM}_N\ge 1/\varepsilon\right\}
        \supseteq
        \left\{K_\pi(\rvLambda_N,\rvmu_N)\ge 1/\varepsilon\right\},
    \end{equation}
    and taking probabilities and using \eqref{eq:mix_markov} gives
    \[
        \mathbb P\!\left[K_\pi(\rvLambda_N,\rvmu_N)\ge 1/\varepsilon\right]\le \varepsilon.
    \]
\end{proof}


We note that the full process $\widetilde{\rvM}_n:=\int_0^\infty \rvM_n(\lambda)\,\pi(d\lambda)$, $n=0,1,\dots,N$, is also a martingale~\cite{kaufmann2021mixture}.

\begin{corollary}\label{cor:mix_complement}
    \(\mathbb P\!\left[K_\pi(N-\rvLambda_N,N-\rvmu_N)\ge 1/\varepsilon\right]\le \varepsilon.\)
\end{corollary}
\begin{proof}
    The complemented variables $\bar{\rvxi}_i:=1-\rvxi_i$ satisfy $\bar{\rvxi}_i\in[0,1]$ and $\mathbb E[\bar{\rvxi}_i\mid\mathcal F_{i-1}]=1-\rvp_i$, so Proposition~\ref{thm:mix_joint_tail} applies to $(\bar{\rvxi}_i,1-\rvp_i)$ with $\bar\rvLambda_N=N-\rvLambda_N$ and $\bar\rvmu_N=N-\rvmu_N$.
\end{proof}

Proposition~\ref{thm:mix_joint_tail} and Corollary~\ref{cor:mix_complement} control the events where $K_\pi$ or its complemented version exceeds $1/\varepsilon$.
To turn these into bounds for $\rvmu_N$ and $\rvLambda_N$, we need the following monotonicity properties of $K_\pi$.
\begin{lemma}\label{lem:mix_gpi_monotonicity}
    For every $s,u\in[0,N]$, the following monotonicity properties hold:
    \begin{align}
        s &\longmapsto K_\pi(s,u)
          && \text{ is nonincreasing on }[0,N], \\
        s &\longmapsto K_\pi(N-s,N-u)
          && \text{ is nondecreasing on }[0,N], \\
        u &\longmapsto K_\pi(s,u)
          && \text{ is nondecreasing on }[0,N], \\
        u &\longmapsto K_\pi(N-s,N-u)
          && \text{ is nonincreasing on }[0,N]. 
    \end{align}
\end{lemma}

\begin{proof}
    We have
    \begin{equation}
        \frac{\partial}{\partial s}e^{-\lambda s}=-\lambda e^{-\lambda s}\le 0,
    \end{equation}
    so the first factor of $K_\pi(s,u)$ is nonincreasing in $s$.
    Also,
    \begin{equation}
        \begin{aligned}
            &\frac{\partial}{\partial u}
            \left[
                1-(1-e^{-\lambda})\frac{u}{N}
            \right]^{-N}
            \\
             &\quad =
            (1-e^{-\lambda})
            \left[
                1-(1-e^{-\lambda})\frac{u}{N}
            \right]^{-N-1}
            \\
             &\quad \ge 0,
        \end{aligned}
    \end{equation}
so the second factor in $K_\lambda(s,u)$ is nondecreasing in $u$.
    Integrating with respect to the probability measure $\pi(d\lambda)$ preserves these monotonicity properties, which proves the first and third claims.
    Replacing $(s,u)$ by $(N-s,N-u)$ then gives the second and fourth claims.
\end{proof}

With these monotonicity properties, each $K_\pi>1/\varepsilon$ event can be inverted in the variable to be bounded.
This gives the following numerical bounds.

\begin{proposition}\label{thm:mix_general_implicit_bounds}
    For $s,u\in[0,N]$, define the boundaries
    \begin{align}
        U_{\pi}^{\mu}(s)
         &:= \inf\!\Bigl[
        \{u\in[0,N]:
        \nonumber \\
         &\qquad
        K_\pi(s,u)> 1/\varepsilon\}
        \cup\{N\}\Bigr]
        \label{eq:mix_general_uimp_mu}     \\
        L_{\pi}^{\mu}(s)
         &:= \sup\!\Bigl[
        \{u\in[0,N]:
        \nonumber \\
         &\qquad
        K_\pi(N-s,N-u)> 1/\varepsilon\}
        \cup\{0\}\Bigr]
        \label{eq:mix_general_limp_mu}     \\
        U_{\pi}^{\Lambda}(u)
         &:= \inf\!\Bigl[
        \{s\in[0,N]:
        \nonumber \\
         &\qquad
        K_\pi(N-s,N-u)> 1/\varepsilon\}
        \cup\{N\}\Bigr]
        \label{eq:mix_general_uimp_lambda}     \\
        L_{\pi}^{\Lambda}(u)
         &:= \sup\!\Bigl[
        \{s\in[0,N]:
        \nonumber \\
         &\qquad
        K_\pi(s,u)> 1/\varepsilon\}
        \cup\{0\}\Bigr]
        \label{eq:mix_general_limp_lambda}
    \end{align}
    Then
    \begin{align}
        \mathbb P\!\left[\rvmu_N> U_{\pi}^{\mu}(\rvLambda_N)\right]     & \le \varepsilon,
        \label{eq:mix_general_mu_upper}                                                    \\
        \mathbb P\!\left[\rvmu_N< L_{\pi}^{\mu}(\rvLambda_N)\right]     & \le \varepsilon,
        \label{eq:mix_general_mu_lower}                                                    \\
        \mathbb P\!\left[\rvLambda_N> U_{\pi}^{\Lambda}(\rvmu_N)\right] & \le \varepsilon.
        \label{eq:mix_general_lambda_upper}
                                     \\
        \mathbb P\!\left[\rvLambda_N< L_{\pi}^{\Lambda}(\rvmu_N)\right] & \le \varepsilon,
        \label{eq:mix_general_lambda_lower}                   
    \end{align}
\end{proposition}

\begin{proof}
    The proof has two steps for each claim.
    First, the required event inclusion follows by combining the relevant boundary definition with a monotonicity property in Lemma~\ref{lem:mix_gpi_monotonicity}.
    Then, the claim follows either from Proposition~\ref{thm:mix_joint_tail} or from Corollary~\ref{cor:mix_complement}.
    The details are given next.

    \medskip\noindent
    \eqref{eq:mix_general_mu_upper}: The map $u\mapsto K_\pi(s,u)$ is nondecreasing.
    By the definition of the infimum boundary $U_{\pi}^{\mu}$, we have
    \[
        \{\rvmu_N>U_{\pi}^{\mu}(\rvLambda_N)\}
        \subseteq
        \{K_\pi(\rvLambda_N,\rvmu_N)\ge 1/\varepsilon\}.
    \]
    Proposition~\ref{thm:mix_joint_tail} gives the claim.

    \medskip\noindent
    \eqref{eq:mix_general_mu_lower}: The map $u\mapsto K_\pi(N-s,N-u)$ is nonincreasing.
    By the definition of the supremum boundary $L_{\pi}^{\mu}$, we have
    \[
        \{\rvmu_N<L_{\pi}^{\mu}(\rvLambda_N)\}
        \subseteq
        \{K_\pi(N-\rvLambda_N,N-\rvmu_N)\ge 1/\varepsilon\}.
    \]
    Corollary~\ref{cor:mix_complement} gives the claim.

    \medskip\noindent
    \eqref{eq:mix_general_lambda_upper}: The map $s\mapsto K_\pi(N-s,N-u)$ is nondecreasing.
    By the definition of the infimum boundary $U_{\pi}^{\Lambda}$, we have
    \[
        \{\rvLambda_N>U_{\pi}^{\Lambda}(\rvmu_N)\}
        \subseteq
        \{K_\pi(N-\rvLambda_N,N-\rvmu_N)\ge 1/\varepsilon\}.
    \]
    Corollary~\ref{cor:mix_complement} gives the claim.

    \medskip\noindent
    \eqref{eq:mix_general_lambda_lower}: The map $s\mapsto K_\pi(s,u)$ is nonincreasing.
    By the definition of the supremum boundary $L_{\pi}^{\Lambda}$, we have
    \[
        \{\rvLambda_N<L_{\pi}^{\Lambda}(\rvmu_N)\}
        \subseteq
        \{K_\pi(\rvLambda_N,\rvmu_N)\ge 1/\varepsilon\}.
    \]
    Proposition~\ref{thm:mix_joint_tail} gives the claim.

\end{proof}

\begingroup
\suppsection{Tonelli's theorem}
\label{app:measure_theory}

Let $(X,\mathcal A,\mu)$ and $(Y,\mathcal B,\nu)$ be $\sigma$-finite measure spaces.
If $h:X\times Y\to[0,\infty]$ is $(\mathcal A\otimes\mathcal B)$-measurable, Tonelli's theorem gives~\cite{tonelli1909integrazione}
\begin{equation}
    \begin{aligned}
         &
        \int_X
        \left[\int_Y h(x,y)\,\nu(dy)\right]\mu(dx)
        \\
         &\quad =
        \int_Y
        \left[\int_X h(x,y)\,\mu(dx)\right]\nu(dy).
    \end{aligned}
\end{equation}
Both sides are well-defined in $[0,\infty]$ and equal the integral of $h$ over the product space.
In Proposition~\ref{thm:mix_joint_tail}, we use this theorem with
\begin{equation}
    \begin{aligned}
        (X,\mathcal A,\mu)
         & :=([0,\infty),\mathfrak B([0,\infty)),\pi),
        \\
        (Y,\mathcal B,\nu)
         & :=(\Omega,\mathcal F,\mathbb P).
    \end{aligned}
\end{equation}
Here $\mathfrak B([0,\infty))$ denotes the Borel $\sigma$-algebra on $[0,\infty)$~\cite{kechris1995classical}.
Both are probability spaces and hence $\sigma$-finite~\cite{klenke2008probability}.
The integrand used in Eq.~\eqref{eq:mix_tonelli_mean} is
\begin{equation}
    h(\lambda,\omega):=\rvM_N(\lambda)(\omega).
\end{equation}
It is nonnegative by construction and measurable as a finite product of measurable factors.
Thus the Tonelli step used in Eq.~\eqref{eq:mix_tonelli_mean} is justified.

\endgroup

\suppsection{Uniform mixture bounds}\label{app:uniform_interval_bounds}

We now specialize the mixture bounds to a uniform distribution on a finite interval.
Fix $0\le A<B$ and write
\begin{equation}
    w:=B-A.
\end{equation}
Define the  uniform distribution by
\begin{equation}
    \pi_{A,B}(d\lambda)=\frac{1}{w}\mathbf 1_{[A,B]}(\lambda)\,d\lambda,
\end{equation}
The corresponding bounds from Proposition~\ref{thm:mix_general_implicit_bounds} are denoted by
$
    U_{A,B}^{\mu},
    L_{A,B}^{\mu},
    L_{A,B}^{\Lambda} \text{ and }
    U_{A,B}^{\Lambda}.
$
For $s,u\in[0,N]$, write
\begin{equation}
    x:=\frac{s}{N},
    \qquad
    y:=\frac{u}{N}.
\end{equation}
Set
\begin{equation}\label{eq:mix_interval_phi}
    \begin{aligned}
        \phi_{x,y}(\lambda)
         & :=
        -\lambda s
        -N\ln\!\left[1-(1-e^{-\lambda})\frac{u}{N}\right]
        \\
         & =
        -N x\lambda
        -N\ln\!\bigl[(1-y)+y e^{-\lambda}\bigr],
        \qquad \lambda\ge 0.
    \end{aligned}
\end{equation}
Then
\begin{equation}\label{eq:mix_interval_kernel}
    \begin{aligned}
        K_{A,B}(s,u)
        &:=
        \frac{1}{w}
        \int_A^B K_\lambda(s,u)\,d\lambda
        \\
        &=
        \frac{1}{w}
        \int_A^B e^{\phi_{x,y}(\lambda)}\,d\lambda.
    \end{aligned}
\end{equation}

We next lower-bound the interval integral in Eq.~\eqref{eq:mix_interval_kernel}, using the part near the maximizer of $\phi_{x,y}(\lambda)$.
The next two lemmas provide the maximizer and the local lower envelope needed for this estimate.
\begin{lemma}\label{lem:mix_phi_facts}
    Assume $s,u\in[0,N]$ and write $x:=s/N$ and $y:=u/N$.
    Assume $0<y<1$.
    Then the following hold:
    \begin{enumerate}
        \item The first and second derivatives of $\phi_{x,y}(\lambda)$ are
              \begin{align}
                  \phi_{x,y}'(\lambda)
                   & =
                  -N x + N\frac{y e^{-\lambda}}{(1-y)+y e^{-\lambda}},
                  \label{eq:mix_phi_prime} \\
                  \phi_{x,y}''(\lambda)
                   & =
                  -N\,\widehat p_\lambda(1-\widehat p_\lambda),
                  \qquad
                  \widehat p_\lambda:=\frac{y e^{-\lambda}}{(1-y)+y e^{-\lambda}}.
                  \label{eq:mix_phi_second}
              \end{align}
        \item The function $\phi_{x,y}(\lambda)$ is strictly concave on $[0,\infty)$.
        \item If $0<x\le y<1$, then the unique stationary point of $\phi_{x,y}(\lambda)$ on $\mathbb R$ is
              \begin{equation}\label{eq:mix_lambda_star}
                  \lambda^*(x,y):=
                  \ln\!\left[\frac{y(1-x)}{x(1-y)}\right],
              \end{equation}
              and this point satisfies $\lambda^*(x,y)\ge 0$.
              \item At that maximizer,
              \begin{equation}\label{eq:mix_phi_at_lambda_star}
                  \phi_{x,y}\!\left[\lambda^*(x,y)\right]=N\,D(x\|y),
              \end{equation}
              where
              \begin{equation}
                  D(x\|y):=
                  x\ln\!\left(\frac{x}{y}\right)+(1-x)\ln\!\left(\frac{1-x}{1-y}\right).
              \end{equation}
    \end{enumerate}
\end{lemma}

\begin{proof}
    Item 1 follows by direct differentiation.
    Since $0<y<1$, we have $\widehat p_\lambda\in(0,1)$ for every $\lambda\ge 0$, and therefore Item 2 follows from \eqref{eq:mix_phi_second}.
    For Item 3, solving $\phi_{x,y}'(\lambda)=0$ using \eqref{eq:mix_phi_prime} gives
    \begin{equation}
        x=\frac{y e^{-\lambda}}{(1-y)+y e^{-\lambda}},
    \end{equation}
    which rearranges to \eqref{eq:mix_lambda_star}; the condition $x\le y$ implies $\lambda^*(x,y)\ge 0$.
    Finally, substituting \eqref{eq:mix_lambda_star} into $\phi_{x,y}(\lambda)$ yields
    \begin{equation}
        \begin{aligned}
            e^{-\lambda^*(x,y)}
             & =
            \frac{x(1-y)}{y(1-x)},
            \\
            (1-y)+y e^{-\lambda^*(x,y)}
             & =
            \frac{1-y}{1-x},
        \end{aligned}
    \end{equation}
    and hence
    \begin{equation}
        \begin{aligned}
            \phi_{x,y}\!\left[\lambda^*(x,y)\right]
             & =
            Nx\ln\!\left[\frac{x(1-y)}{y(1-x)}\right]
            -N\ln\!\left[\frac{1-y}{1-x}\right]
            \\
             & =
            N\,D(x\|y),
        \end{aligned}
    \end{equation}
    which is Item 4.
\end{proof}

\begin{lemma}\label{lem:mix_curvature}
    Assume $0<y<1$.
    Then
    \begin{equation}\label{eq:mix_curvature_bound}
        -\phi_{x,y}''(\lambda)\le \frac{N}{4}
        \qquad\text{for every }\lambda\ge 0.
    \end{equation}
    Consequently, if $0<x\le y<1$ and $\lambda^*(x,y)\in[A,B]$, then
    \begin{equation}\label{eq:mix_quadratic_envelope}
        \begin{aligned}
            \phi_{x,y}(\lambda)
             & \ge
            \phi_{x,y}\!\left[\lambda^*(x,y)\right]
            -\frac{N}{8}
            \bigl[\lambda-\lambda^*(x,y)\bigr]^2
        \end{aligned}
    \end{equation}
for every $\lambda\in[A,B]$.
\end{lemma}

\begin{proof}
    By \eqref{eq:mix_phi_second},
    \begin{equation}
        -\phi_{x,y}''(\lambda)=N\,\widehat p_\lambda(1-\widehat p_\lambda).
    \end{equation}
    Since $\widehat p_\lambda\in[0,1]$ and $\widehat p_\lambda(1-\widehat p_\lambda)\le 1/4$, \eqref{eq:mix_curvature_bound} follows immediately.

    Now assume $0<x\le y<1$ and write $\lambda^*:=\lambda^*(x,y)$.
    Taylor's theorem~\cite{apostol1974mathematicalanalysis} with integral remainder around $\lambda^*$ gives, for every $\lambda\in[A,B]$,
    \begin{equation}
        \begin{aligned}
            \phi_{x,y}(\lambda)
             & =
            \phi_{x,y}(\lambda^*)
            +(\lambda-\lambda^*)\phi_{x,y}'(\lambda^*)
            \\
             &\quad
            +(\lambda-\lambda^*)^2
            \int_0^1 (1-t)
            \phi_{x,y}''\!\bigl[\lambda^*+t(\lambda-\lambda^*)\bigr]
            \,dt.
        \end{aligned}
    \end{equation}
    Because $\phi_{x,y}'(\lambda^*)=0$ and $\phi_{x,y}''\ge -N/4$ by \eqref{eq:mix_curvature_bound}, we obtain
    \begin{align}
        \phi_{x,y}(\lambda)
         & \ge
        \phi_{x,y}(\lambda^*)
        -\frac{N}{4}(\lambda-\lambda^*)^2\int_0^1(1-t)\,dt\nonumber \\
         & =
        \phi_{x,y}(\lambda^*)-\frac{N}{8}(\lambda-\lambda^*)^2,
    \end{align}
    which is \eqref{eq:mix_quadratic_envelope}.
\end{proof}

We now apply the lower bound in Eq.~\eqref{eq:mix_quadratic_envelope} to the integral representation in Eq.~\eqref{eq:mix_interval_kernel}.
This allows us to relate the threshold $K_{A,B}(Nx,Ny)\ge1/\varepsilon$ to the KL divergence, as shown next.
For this analytical comparison we assume $\lambda^*(x,y)\in[A,B]$, but we remark that the numerical inversion of $K_{A,B}$ remains valid for any fixed interval $[A,B]$, even when $\lambda^*(x,y)$ lies outside it.

\begin{proposition}\label{prop:mix_interval_width_only}
    Assume $0<x\le y<1$, $\lambda^*(x,y)\in[A,B]$, and $w:=B-A$.
    Let
    \begin{equation}\label{eq:mix_effective_epsilon_def}
        \widetilde{\varepsilon}_{A,B}
        :=
        \varepsilon
        \left[1+\frac{Nw^2}{2\pi}\right]^{-1/2}.
    \end{equation}
    Then 
    \begin{equation}\label{eq:mix_interval_kl_log_only}
        N\,D(x\|y)\ge \ln(1/\widetilde{\varepsilon}_{A,B})
        \quad\Longrightarrow\quad
        K_{A,B}(Nx,Ny)\ge 1/\varepsilon.
    \end{equation}
\end{proposition}

\begin{proof}
Using \eqref{eq:mix_interval_kernel} and \eqref{eq:mix_quadratic_envelope}, we have that
    \begin{equation}
        \begin{aligned}
            K_{A,B}(Nx,Ny)
             & =
            \frac{1}{w}\int_A^B e^{\phi_{x,y}(\lambda)}\,d\lambda
            \\
             & \ge
            \frac{e^{\phi_{x,y}(\lambda^*)}}{w}
            \\
             &\quad \times
            \int_A^B
            e^{-\frac{N}{8}(\lambda-\lambda^*)^2}d\lambda,
        \end{aligned}
    \end{equation}
    where $\lambda^*=\lambda^*(x,y)$.
    It remains to lower-bound the mass over $[A,B]$ of the Gaussian curve peaked at $\lambda^*$.
    Since this peak lies inside $[A,B]$, symmetry and monotonicity away from the peak imply that the captured mass is minimized when the peak is at an endpoint.
    Therefore
    \begin{equation}
        \begin{aligned}
            &\int_A^B e^{-N(\lambda-\lambda^*)^2/8}\,d\lambda
            \\
            &\quad\ge
            \int_A^B e^{-N(\lambda-A)^2/8}\,d\lambda
            =
            \int_0^w e^{-Nv^2/8}\,dv .
        \end{aligned}
    \end{equation}
    Using this and \eqref{eq:mix_phi_at_lambda_star} gives
    \begin{equation}
        K_{A,B}(Nx,Ny)
        \ge
        \frac{e^{N D(x\|y)}}{w}
        \int_0^w e^{-\frac{N}{8}v^2}dv.
    \end{equation}
    The substitution $z=v\sqrt{N/8}$ gives $dv=\sqrt{8/N}\,dz$ and yields
    \begin{equation}
        \begin{aligned}
            \int_0^w e^{-\frac{N}{8}v^2}dv
             & =
            \sqrt{\frac{8}{N}}
            \int_0^{w\sqrt{N/8}} e^{-z^2}\,dz
            \\
             & =
            \sqrt{\frac{2\pi}{N}}\,
            \operatorname{erf}\!\left(w\sqrt{\frac{N}{8}}\right),
        \end{aligned}
    \end{equation}
    and hence
    \begin{equation}\label{eq:mix_interval_width_only}
        K_{A,B}(Nx,Ny)
        \ge
        \frac{e^{N D(x\|y)}}{w}
        \sqrt{\frac{2\pi}{N}}\,
        \operatorname{erf}\!\left(w\sqrt{\frac{N}{8}}\right).
    \end{equation}
    The log-quadratic Gaussian-tail bound of Ref.~\cite{mastin_jaillet2013logquadratic} gives, for $z\ge0$,
    \begin{equation}
        \operatorname{erf}(z)
        \ge
        1-\exp\!\left(
            -\frac{2z^2}{\pi}
            -\frac{2z}{\sqrt{\pi}}
        \right).
    \end{equation}
    Together with $e^\theta\ge 1+\theta+\theta^2/2$ for $\theta\ge0$, this implies
    \begin{equation}\label{eq:erf_log_only_lower_bound}
        \operatorname{erf}(z)
        \ge
        \frac{2z}{\sqrt{\pi+4z^2}},
        \qquad z\ge 0.
    \end{equation}
    Applying \eqref{eq:erf_log_only_lower_bound} with $z=w\sqrt{N/8}$ gives
    \begin{equation}
        \frac{1}{w}\sqrt{\frac{2\pi}{N}}\,
        \operatorname{erf}\!\left(w\sqrt{\frac{N}{8}}\right)
        \ge
        \left(1+\frac{Nw^2}{2\pi}\right)^{-1/2}.
    \end{equation}
    Combining this with \eqref{eq:mix_interval_width_only} yields
    \begin{equation}\label{eq:mix_interval_width_only_log_bound}
        K_{A,B}(Nx,Ny)
        \ge
        e^{ND(x\|y)}
        \left(1+\frac{Nw^2}{2\pi}\right)^{-1/2}.
    \end{equation}
    Therefore \eqref{eq:mix_interval_kl_log_only} follows from \eqref{eq:mix_effective_epsilon_def}.
\end{proof}

\subsection{Closed-form rational bounds}
We now obtain relaxed but tight analytical bounds by applying the KL relaxation of Refs.~\cite{topsoe2007some,mannalath2025sharp} to the sufficient condition in Eq.~\eqref{eq:mix_interval_kl_log_only}.
\begin{lemma}\label{lem:mix_surrogates}
    Let $w:=B-A$ and define
    \begin{equation}
        \widetilde{\varepsilon}_{A,B}
        :=
        \varepsilon
        \left(1+\frac{Nw^2}{2\pi}\right)^{-1/2},
        \qquad
	        \widetilde{\kappa}:=\frac{2\ln(1/\widetilde{\varepsilon}_{A,B})}{9N}.
	    \end{equation}
	    The angle-bracket labels $\langle\mu\rangle$ and $\langle\Lambda\rangle$ indicate the later bounds in which these quantities are used (e.g., $\widetilde{\gamma}_{\langle\mu\rangle}^{\pm}$ enters the bounds on $\rvmu_N$, whereas $\widetilde{\zeta}_{\langle\Lambda\rangle}^{\pm}$ enters the bounds on $\rvLambda_N$), not a functional dependence.
	    Define
    \begin{equation}
        \begin{aligned}
            \widetilde{a}_{\langle\mu\rangle}
             & :=1+4\widetilde{\kappa},
            \\
            \widetilde{b}_{\langle\mu\rangle}(x)
             & :=-6\widetilde{\kappa}-2(1-2\widetilde{\kappa})x,
            \\
            \widetilde{c}_{\langle\mu\rangle}(x)
             & :=(1+\widetilde{\kappa})x^2-3\widetilde{\kappa}x,
        \end{aligned}
    \end{equation}
    \begin{equation}
        \begin{aligned}
            \widetilde{\gamma}_{\langle\mu\rangle}^{\pm}(x)
             &:=
            \frac{
                -\widetilde{b}_{\langle\mu\rangle}(x)
                \pm
                \sqrt{
                    \widetilde{b}_{\langle\mu\rangle}(x)^2
                    -4\widetilde{a}_{\langle\mu\rangle}
                    \widetilde{c}_{\langle\mu\rangle}(x)
                }
            }
            {2\widetilde{a}_{\langle\mu\rangle}}.
        \end{aligned}
        \label{eq:mix_gamma_def}
    \end{equation}
    \begin{equation}
        \widetilde{r}_{\langle\mu\rangle}^{-}
        :=
        \frac{3\widetilde{\kappa}}{1+\widetilde{\kappa}},
        \qquad
        \widetilde{r}_{\langle\mu\rangle}^{+}
        :=
        \frac{1-2\widetilde{\kappa}}{1+\widetilde{\kappa}},
    \end{equation}
    \begin{equation}
        \begin{aligned}
            \widetilde{a}_{\langle\Lambda\rangle}
             & :=1+\widetilde{\kappa},
            \\
            \widetilde{b}_{\langle\Lambda\rangle}(y)
             & :=-3\widetilde{\kappa}-2(1-2\widetilde{\kappa})y,
            \\
            \widetilde{c}_{\langle\Lambda\rangle}(y)
             & :=(1+4\widetilde{\kappa})y^2-6\widetilde{\kappa}y,
        \end{aligned}
    \end{equation}
    \begin{equation}
        \begin{aligned}
            \widetilde{\zeta}_{\langle\Lambda\rangle}^{\pm}(y)
             &:=
            \frac{
                -\widetilde{b}_{\langle\Lambda\rangle}(y)
                \pm
                \sqrt{
                    \widetilde{b}_{\langle\Lambda\rangle}(y)^2
                    -4\widetilde{a}_{\langle\Lambda\rangle}
                    \widetilde{c}_{\langle\Lambda\rangle}(y)
                }
            }
            {2\widetilde{a}_{\langle\Lambda\rangle}}.
        \end{aligned}
        \label{eq:mix_zeta_def}
    \end{equation}
    \begin{equation}
        \widetilde{r}_{\langle\Lambda\rangle}^{+}
        :=
        \frac{1-2\widetilde{\kappa}}{1+4\widetilde{\kappa}},
        \qquad
        \widetilde{r}_{\langle\Lambda\rangle}^{-}
        :=
        \frac{6\widetilde{\kappa}}{1+4\widetilde{\kappa}}.
    \end{equation}
    Then for $x,y\in(0,1)$, each displayed condition implies the bound:
    \begingroup
    \footnotesize
    \begin{equation}
        \left.
        \begin{aligned}
            x\in(0,\widetilde{r}_{\langle\mu\rangle}^{+}),\quad
            y\ge \widetilde{\gamma}_{\langle\mu\rangle}^{+}(x),\\
            x\in(\widetilde{r}_{\langle\mu\rangle}^{-},1),\quad
            y\le \widetilde{\gamma}_{\langle\mu\rangle}^{-}(x),\\
            y\in(\widetilde{r}_{\langle\Lambda\rangle}^{-},1),\quad
            x\le \widetilde{\zeta}_{\langle\Lambda\rangle}^{-}(y),\\
            y\in(0,\widetilde{r}_{\langle\Lambda\rangle}^{+}),\quad
            x\ge \widetilde{\zeta}_{\langle\Lambda\rangle}^{+}(y)
        \end{aligned}
        \right\}
        \Longrightarrow
        D(x\|y)\ge \frac{1}{N}\ln(1/\widetilde{\varepsilon}_{A,B}).
        \label{eq:mix_surrogate_cert}
    \end{equation}
    \endgroup
\end{lemma}

\begin{proposition}\label{thm:mix_explicit_bounds}
    For $s,u\in[0,N]$, write $x:=s/N$ and $y:=u/N$.
    Using the functions $\widetilde{\gamma}_{\langle\mu\rangle}^{\pm}$ and $\widetilde{\zeta}_{\langle\Lambda\rangle}^{\pm}$ and the endpoints $\widetilde{r}_{\langle\mu\rangle}^{\pm}$ and $\widetilde{r}_{\langle\Lambda\rangle}^{\pm}$ from Lemma~\ref{lem:mix_surrogates}, define
    \begingroup
    \footnotesize
    \begin{align}
        U_{A,B}^{\mu,\mathrm{rat}}(s)
         & :=
        \begin{cases}
            N\,\widetilde{\gamma}_{\langle\mu\rangle}^+(x),
             &
            \begin{aligned}[t]
                &\text{if }x\in(0,\widetilde{r}_{\langle\mu\rangle}^{+}),
                \\
                &\lambda^*\!\left[x,\widetilde{\gamma}_{\langle\mu\rangle}^+(x)\right]\in[A,B],
            \end{aligned}
            \\[1ex]
            N,
             & \text{otherwise,}
        \end{cases}     \\
        L_{A,B}^{\mu,\mathrm{rat}}(s)
         & :=
        \begin{cases}
            N\,\widetilde{\gamma}_{\langle\mu\rangle}^-(x),
             &
            \begin{aligned}[t]
                &\text{if }x\in(\widetilde{r}_{\langle\mu\rangle}^{-},1),
                \\
                &\lambda^*\!\left[1-x,1-\widetilde{\gamma}_{\langle\mu\rangle}^-(x)\right]\in[A,B],
            \end{aligned}
            \\[1ex]
            0,
             & \text{otherwise,}
        \end{cases} \\
        U_{A,B}^{\Lambda,\mathrm{rat}}(u)
         & :=
        \begin{cases}
            N\,\widetilde{\zeta}_{\langle\Lambda\rangle}^+(y),
             &
            \begin{aligned}[t]
                &\text{if }y\in(0,\widetilde{r}_{\langle\Lambda\rangle}^{+}),
                \\
                &\lambda^*\!\left[1-\widetilde{\zeta}_{\langle\Lambda\rangle}^+(y),1-y\right]\in[A,B],
            \end{aligned}
            \\[1ex]
            N,
             & \text{otherwise,}
        \end{cases} \\
        L_{A,B}^{\Lambda,\mathrm{rat}}(u)
         & :=
        \begin{cases}
            N\,\widetilde{\zeta}_{\langle\Lambda\rangle}^-(y),
             &
            \begin{aligned}[t]
                &\text{if }y\in(\widetilde{r}_{\langle\Lambda\rangle}^{-},1),
                \\
                &\lambda^*\!\left[\widetilde{\zeta}_{\langle\Lambda\rangle}^-(y),y\right]\in[A,B],
            \end{aligned}
            \\[1ex]
            0,
             & \text{otherwise.}
        \end{cases}
    \end{align}
    \endgroup
    Then
    \begin{align}
        \mathbb P\!\left[\rvmu_N> U_{A,B}^{\mu,\mathrm{rat}}(\rvLambda_N)\right]     & \le \varepsilon,
        \label{eq:mix_final_mu_upper}                                                                 \\
        \mathbb P\!\left[\rvmu_N< L_{A,B}^{\mu,\mathrm{rat}}(\rvLambda_N)\right]     & \le \varepsilon,
        \label{eq:mix_final_mu_lower}                                                                 \\
        \mathbb P\!\left[\rvLambda_N> U_{A,B}^{\Lambda,\mathrm{rat}}(\rvmu_N)\right] & \le \varepsilon,
        \label{eq:mix_final_lambda_upper}                                                             \\
        \mathbb P\!\left[\rvLambda_N< L_{A,B}^{\Lambda,\mathrm{rat}}(\rvmu_N)\right] & \le \varepsilon.
        \label{eq:mix_final_lambda_lower}
    \end{align}
\end{proposition}

\begin{proof}
    It is enough to consider the nontrivial branch in each definition.
    There, \eqref{eq:mix_interval_kl_log_only} gives $K_{A,B}(s,u)\ge 1/\varepsilon$ on the boundary.
	    The monotonicity established in Lemma~\ref{lem:mix_gpi_monotonicity} extends this to the whole tail region, and Proposition~\ref{thm:mix_joint_tail} or Corollary~\ref{cor:mix_complement} gives the required probability bound.

    \medskip\noindent
    \eqref{eq:mix_final_mu_upper}: \eqref{eq:mix_surrogate_cert} gives $D\!\left[x\|\widetilde{\gamma}_{\langle\mu\rangle}^+(x)\right]\ge \ln(1/\widetilde{\varepsilon}_{A,B})/N$, so \eqref{eq:mix_interval_kl_log_only} yields $K_{A,B}\!\left[s,N\widetilde{\gamma}_{\langle\mu\rangle}^+(x)\right]\ge 1/\varepsilon$. Since $u\mapsto K_{A,B}(s,u)$ is nondecreasing, $\{\rvmu_N>U_{A,B}^{\mu,\mathrm{rat}}(\rvLambda_N)\}\subseteq\{K_{A,B}(\rvLambda_N,\rvmu_N)\ge 1/\varepsilon\}$. Proposition~\ref{thm:mix_joint_tail} gives the claim.

    \medskip\noindent
    \eqref{eq:mix_final_mu_lower}: \eqref{eq:mix_surrogate_cert} gives $D\!\left[x\|\widetilde{\gamma}_{\langle\mu\rangle}^-(x)\right]\ge \ln(1/\widetilde{\varepsilon}_{A,B})/N$.
    Since Bernoulli KL divergence is invariant under complementation,
    \[
        D\!\left[1-x\|1-\widetilde{\gamma}_{\langle\mu\rangle}^-(x)\right]
        =
        D\!\left[x\|\widetilde{\gamma}_{\langle\mu\rangle}^-(x)\right].
    \]
    Thus \eqref{eq:mix_interval_kl_log_only} yields $K_{A,B}\!\left[N-s,N-N\widetilde{\gamma}_{\langle\mu\rangle}^-(x)\right]\ge 1/\varepsilon$.
    Since $u\mapsto K_{A,B}(N-s,N-u)$ is nonincreasing, $\{\rvmu_N<L_{A,B}^{\mu,\mathrm{rat}}(\rvLambda_N)\}\subseteq\{K_{A,B}(N-\rvLambda_N,N-\rvmu_N)\ge 1/\varepsilon\}$. Corollary~\ref{cor:mix_complement} gives the claim.

    \medskip\noindent
    \eqref{eq:mix_final_lambda_upper}: \eqref{eq:mix_surrogate_cert} gives $D\!\left[\widetilde{\zeta}_{\langle\Lambda\rangle}^+(y)\|y\right]\ge \ln(1/\widetilde{\varepsilon}_{A,B})/N$.
    Since Bernoulli KL divergence is invariant under complementation,
    \[
        D\!\left[1-\widetilde{\zeta}_{\langle\Lambda\rangle}^+(y)\|1-y\right]
        =
        D\!\left[\widetilde{\zeta}_{\langle\Lambda\rangle}^+(y)\|y\right].
    \]
    Thus \eqref{eq:mix_interval_kl_log_only} yields $K_{A,B}\!\left[N-N\widetilde{\zeta}_{\langle\Lambda\rangle}^+(y),N-u\right]\ge 1/\varepsilon$.
    Since $s\mapsto K_{A,B}(N-s,N-u)$ is nondecreasing, $\{\rvLambda_N>U_{A,B}^{\Lambda,\mathrm{rat}}(\rvmu_N)\}\subseteq\{K_{A,B}(N-\rvLambda_N,N-\rvmu_N)\ge 1/\varepsilon\}$. Corollary~\ref{cor:mix_complement} gives the claim.

    \medskip\noindent
    \eqref{eq:mix_final_lambda_lower}: \eqref{eq:mix_surrogate_cert} gives $D\!\left[\widetilde{\zeta}_{\langle\Lambda\rangle}^-(y)\|y\right]\ge \ln(1/\widetilde{\varepsilon}_{A,B})/N$.
    Thus \eqref{eq:mix_interval_kl_log_only} yields $K_{A,B}\!\left[N\widetilde{\zeta}_{\langle\Lambda\rangle}^-(y),u\right]\ge 1/\varepsilon$.
    Since $s\mapsto K_{A,B}(s,u)$ is nonincreasing, $\{\rvLambda_N<L_{A,B}^{\Lambda,\mathrm{rat}}(\rvmu_N)\}\subseteq\{K_{A,B}(\rvLambda_N,\rvmu_N)\ge 1/\varepsilon\}$.
    Proposition~\ref{thm:mix_joint_tail} gives the claim.
\end{proof}

In summary, the analytical bounds of Ref.~\cite{mannalath2025sharp}, which control the departure of a sum of independent Bernoulli trials from its mean, lift to the non-independent and identically distributed (non-IID) setting, where they control the departure of the observed sum from the sum of conditional expectations.
The price of this lift is the mixture penalty $\frac12\ln[1+N(B-A)^2/(2\pi)]$, paid to obtain robustness over the interval of martingale parameters.

\suppsection{\texorpdfstring{Fixed-$\lambda$ martingale bounds}{Fixed-lambda closed forms and guess calibration}}
\label{app:fixed_lambda_closed_forms}

Taking $\pi=\delta_\lambda$ in Proposition~\ref{thm:mix_general_implicit_bounds} gives $K_\pi=K_\lambda$ and recovers the fixed-$\lambda$ martingale bounds.
The following proposition states the resulting closed forms.
\begin{proposition}\label{prop:fixed_lambda_bounds}
    Fix $\lambda>0$ and $\varepsilon\in(0,1)$.
    For $s,u\in[0,N]$, define
\begin{equation}
    K_\lambda(s,u)
    =
    e^{-\lambda s}
    \left[
        1-(1-e^{-\lambda})\frac{u}{N}
    \right]^{-N}.
\end{equation}
Then
\begin{equation}
    \mathbb P\!\left[
        K_\lambda(\rvLambda_N,\rvmu_N)\ge 1/\varepsilon
    \right]\le \varepsilon.
\end{equation}
Equivalently, we have that
\begin{equation}
    \mathbb P\!\left[
        -\lambda\rvLambda_N
        -N\ln\!\left(
            1-(1-e^{-\lambda})\frac{\rvmu_N}{N}
        \right)
        \ge \ln(1/\varepsilon)
    \right]\le \varepsilon.
\end{equation}
The upper-$\rvmu_N$ and lower-$\rvLambda_N$ maps are obtained by solving $K_\lambda(s,u)=1/\varepsilon$ in $u$ and $s$, respectively.
The lower-$\rvmu_N$ and upper-$\rvLambda_N$ maps are obtained by applying the same two inversions to $K_\lambda(N-s,N-u)=1/\varepsilon$.
Define the clipping map $\Pi_N(z):=\min\{N,\max\{0,z\}\}$.
Define
\begin{align}
    U_{\lambda}^{\mu}(s)
     &:=
    \Pi_N\!\left(
    \frac{
        N\!\left[
            1-e^{\frac{-\lambda s-\ln(1/\varepsilon)}{N}}
        \right]
    }
    {1-e^{-\lambda}}
    \right),
    \label{eq:fixed_lambda_u_mu}
    \\
    L_{\lambda}^{\mu}(s)
     &:=
    \Pi_N\!\left(
    \frac{
        N\!\left[
            1-e^{\frac{\lambda s-\ln(1/\varepsilon)}{N}}
        \right]
    }
    {1-e^{\lambda}}
    \right),
    \label{eq:fixed_lambda_l_mu}
    \\
    U_{\lambda}^{\Lambda}(u)
     &:=
    \Pi_N\!\left(
    \frac{
        \ln(1/\varepsilon)
        +N\ln\!\left[
            1-(1-e^{\lambda})\frac{u}{N}
        \right]
    }
    {\lambda}
    \right),
    \label{eq:fixed_lambda_u_lambda}
    \\
    L_{\lambda}^{\Lambda}(u)
     &:=
    \Pi_N\!\left(
    \frac{
        -\ln(1/\varepsilon)
        -N\ln\!\left[
            1-(1-e^{-\lambda})\frac{u}{N}
        \right]
    }
    {\lambda}
    \right).
    \label{eq:fixed_lambda_l_lambda}
\end{align}
Then
\begin{align}
    \mathbb P\!\left[\rvmu_N>U_{\lambda}^{\mu}(\rvLambda_N)\right]
     &\le
    \varepsilon,
    &
    \mathbb P\!\left[\rvmu_N<L_{\lambda}^{\mu}(\rvLambda_N)\right]
     &\le
    \varepsilon,
    \nonumber\\
    \mathbb P\!\left[\rvLambda_N>U_{\lambda}^{\Lambda}(\rvmu_N)\right]
     &\le
    \varepsilon,
    &
    \mathbb P\!\left[\rvLambda_N<L_{\lambda}^{\Lambda}(\rvmu_N)\right]
     &\le
    \varepsilon.
    \label{eq:fixed_lambda_four_valid}
\end{align}
\end{proposition}

We first identify, for fixed deterministic inputs $s$ and $u$, the values of $\lambda$ that optimize its four bounds. 
For $s,u\in(0,N)$, on branches where clipping is inactive and the bounds are nontrivial, the optimization programs are
\begin{align}
    \lambda_{\mu,+}(s)
     &\in \arg\min_{\lambda>0} U_\lambda^\mu(s),
    &
    \lambda_{\mu,-}(s)
     &\in \arg\max_{\lambda>0} L_\lambda^\mu(s),
    \nonumber\\
    \lambda_{\Lambda,+}(u)
     &\in \arg\min_{\lambda>0} U_\lambda^\Lambda(u),
    &
    \lambda_{\Lambda,-}(u)
     &\in \arg\max_{\lambda>0} L_\lambda^\Lambda(u).
    \label{eq:fixed_lambda_calibration_programs}
\end{align}
Writing $x:=s/N$ and $y:=u/N$, the solutions of Eq.~\eqref{eq:fixed_lambda_calibration_programs} can be expressed through exact KL boundaries, defined by the equations
\begin{equation}
    \begin{aligned}
        D\!\left[x\middle\|y_\pm(x)\right]
         &=\frac{\ln(1/\varepsilon)}{N},
        & y_-(x)<x<y_+(x),
        \\
        D\!\left[x_\pm(y)\middle\|y\right]
         &=\frac{\ln(1/\varepsilon)}{N},
        & x_-(y)<y<x_+(y).
    \end{aligned}
    \label{eq:fixed_lambda_exact_kl_boundaries}
\end{equation}
The exact optimizers of Eq.~\eqref{eq:fixed_lambda_calibration_programs} are then given by

\renewcommand{\minalignsep}{0pt}
\begin{align}
    \lambda_{\mu,+}(s) &= \lambda^*\!\left[x,y_+(x)\right], \nonumber\\
    \lambda_{\mu,-}(s) &= \lambda^*\!\left[1-x,1-y_-(x)\right], \nonumber\\
    \lambda_{\Lambda,+}(u) &= \lambda^*\!\left[1-x_+(y),1-y\right], \nonumber\\
    \lambda_{\Lambda,-}(u) &= \lambda^*\!\left[x_-(y),y\right].
\end{align}

These values would yield the tightest bounds for observed inputs $s$ and $u$. With this in mind, in the same spirit as the Kato bounds, we calibrate $\lambda$ using pre-data guesses $g_\Lambda,g_\mu\in[0,N]$, setting $x_g:=g_\Lambda/N$ and $y_g:=g_\mu/N$.

Solving the exact KL equations in Eq.~\eqref{eq:fixed_lambda_exact_kl_boundaries} at $x_g$ and $y_g$ requires numerical inversion.
To obtain explicit fixed calibrations, we use the rational relaxations of Lemma~\ref{lem:mix_surrogates} with the interval penalty removed, so that $\widetilde{\varepsilon}_{A,B}=\varepsilon$.
Denote the resulting functions and endpoints by $\gamma_{\langle\mu\rangle,\varepsilon}^{\pm}$, $\zeta_{\langle\Lambda\rangle,\varepsilon}^{\pm}$, $r_{\langle\mu\rangle,\varepsilon}^{\pm}$, and $r_{\langle\Lambda\rangle,\varepsilon}^{\pm}$.

Fix a margin $\iota>0$ such that the trimmed intervals below are nonempty.
Set
\begin{align*}
    \bar x_g
     &:=
    \min\!\left\{
    r_{\langle\mu\rangle,\varepsilon}^{+}-\iota,
    \max\!\left\{
    r_{\langle\mu\rangle,\varepsilon}^{-}+\iota,
    x_g
    \right\}
    \right\},
    \\
    \bar y_g
     &:=
    \min\!\left\{
    r_{\langle\Lambda\rangle,\varepsilon}^{+}-\iota,
    \max\!\left\{
    r_{\langle\Lambda\rangle,\varepsilon}^{-}+\iota,
    y_g
    \right\}
    \right\}.
\end{align*}
The margin prevents the endpoint cases where the corresponding value of $\lambda^*$ diverges.
For $0<x\le y<1$, let
\(
    \lambda^*(x,y)
\)
be given by Eq.~\eqref{eq:mix_lambda_star}.
The calibrated parameters are
\begin{align}
    \lambda_{\mu,+}^{\mathrm{rat}}(g_\Lambda)
     &:=
    \lambda^*\!\left[\bar x_g,\gamma_{\langle\mu\rangle,\varepsilon}^+(\bar x_g)\right],
    \label{eq:fixed_lambda_mu_plus_rat}
    \\
    \lambda_{\mu,-}^{\mathrm{rat}}(g_\Lambda)
     &:=
    \lambda^*\!\left[1-\bar x_g,1-\gamma_{\langle\mu\rangle,\varepsilon}^-(\bar x_g)\right],
    \label{eq:fixed_lambda_mu_minus_rat}
    \\
    \lambda_{\Lambda,+}^{\mathrm{rat}}(g_\mu)
     &:=
    \lambda^*\!\left[1-\zeta_{\langle\Lambda\rangle,\varepsilon}^+(\bar y_g),1-\bar y_g\right],
    \label{eq:fixed_lambda_lambda_plus_rat}
    \\
    \lambda_{\Lambda,-}^{\mathrm{rat}}(g_\mu)
     &:=
    \lambda^*\!\left[\zeta_{\langle\Lambda\rangle,\varepsilon}^-(\bar y_g),\bar y_g\right].
    \label{eq:fixed_lambda_lambda_minus_rat}
\end{align}

By Eq.~\eqref{eq:fixed_lambda_four_valid}, the calibrated fixed-$\lambda$ bounds satisfy
\begingroup
\small
\begin{align}
    \mathbb P\!\left[
        \rvmu_N>U_{\lambda_{\mu,+}^{\mathrm{rat}}(g_\Lambda)}^{\mu}(\rvLambda_N)
    \right]
     &\le
    \varepsilon,
    &
    \mathbb P\!\left[
        \rvmu_N<L_{\lambda_{\mu,-}^{\mathrm{rat}}(g_\Lambda)}^{\mu}(\rvLambda_N)
    \right]
     &\le
    \varepsilon,
    \nonumber\\
    \mathbb P\!\left[
        \rvLambda_N>U_{\lambda_{\Lambda,+}^{\mathrm{rat}}(g_\mu)}^{\Lambda}(\rvmu_N)
    \right]
     &\le
    \varepsilon,
    &
    \mathbb P\!\left[
        \rvLambda_N<L_{\lambda_{\Lambda,-}^{\mathrm{rat}}(g_\mu)}^{\Lambda}(\rvmu_N)
    \right]
     &\le
    \varepsilon.
    \label{eq:fixed_calibrated_kl_valid}
\end{align}
\endgroup

\suppsection{Binomial benchmark bounds}
\label{sec:mix_optimality}

Here we prove the binomial benchmark claim used in the main text.
Denote the adaptive Bernoulli class by $\mathcal A_N$ and the IID Bernoulli subclass by $\mathcal B_N$.
Since $\mathcal B_N\subset\mathcal A_N$, any confidence boundary valid uniformly over $\mathcal A_N$ must also be valid for every IID Bernoulli process.
The binomial tail inversion is therefore a benchmark for bounds valid uniformly over the adaptive class.

For $s,u\in[0,N]$, write the lower binomial tail as
\begin{equation}
    B_N(s,u)
    :=
    \Prob\!\left[\operatorname{Bin}(N,u/N)\le s\right].
    \label{eq:opt_bin_tail_def}
\end{equation}

\begin{proposition}\label{prop:binomial-benchmark-sm}
    Fix $\eps\in(0,1)$.
    Let $U^\mu,L^\Lambda:[0,N]\to[0,N]$ be nondecreasing functions satisfying
    \[
        \Prob[\rvmu_N>U^\mu(\rvLambda_N)]\le\eps,
        \qquad
        \Prob[\rvLambda_N<L^\Lambda(\rvmu_N)]\le\eps
    \]
    uniformly over $\mathcal A_N$.
    Define
    \begin{align}
        U^\mu_{\mathrm{Bin}}(s)
         &:=
        \sup\!\left\{v\in[0,N]:B_N(s,v)>\eps\right\},
        \nonumber\\
        L^\Lambda_{\mathrm{Bin}}(u)
         &:=
        \inf\!\left\{r\in[0,N]:B_N(r,u)>\eps\right\}.
        \label{eq:binomial-benchmark-def-sm}
    \end{align}
    Then, for all $s,u\in[0,N]$,
    \begin{equation}
        U^\mu(s)\ge U^\mu_{\mathrm{Bin}}(s),
        \qquad
        L^\Lambda(u)\le L^\Lambda_{\mathrm{Bin}}(u).
        \label{eq:binomial-benchmark-sm}
    \end{equation}
\end{proposition}

\begin{proof}
    It suffices to prove the following two implications:
    \begin{equation}
        \begin{aligned}
            B_N(s,v)>\eps
            \quad&\Longrightarrow\quad
            U^\mu(s)\ge v,
            \\
            B_N(r,u)>\eps
            \quad&\Longrightarrow\quad
            L^\Lambda(u)\le r.
        \end{aligned}
        \label{eq:binomial-benchmark-proof-targets}
    \end{equation}
    Taking the supremum over admissible $v$ and the infimum over admissible $r$ then gives \eqref{eq:binomial-benchmark-sm}.

    We first prove the upper-$\rvmu_N$ implication.
    Fix $s\in[0,N]$ and choose $v\in[0,N]$ such that $B_N(s,v)>\eps$.
    Consider the IID Bernoulli$(v/N)$ process.
    This process belongs to $\mathcal B_N\subset\mathcal A_N$, has $\rvmu_N=v$ deterministically, and gives $\rvLambda_N\sim\operatorname{Bin}(N,v/N)$.
    If $U^\mu(s)<v$, then monotonicity of $U^\mu$ gives $U^\mu(\rvLambda_N)\le U^\mu(s)<v=\rvmu_N$ on the event $\{\rvLambda_N\le s\}$.
    Therefore
    \begin{equation}
        \begin{aligned}
            \Prob\!\left[\rvmu_N>U^\mu(\rvLambda_N)\right]
             &\ge
            \Prob\!\left[\rvLambda_N\le s\right]
            \\
             &=
            B_N(s,v)
            >\eps,
        \end{aligned}
    \end{equation}
    This contradicts the assumed upper-$\rvmu_N$ validity.
    Hence $U^\mu(s)\ge v$ for every such $v$.

    We next prove the lower-$\rvLambda_N$ implication.
    Fix $u\in[0,N]$ and take any $r\in[0,N]$ such that $B_N(r,u)>\eps$.
    Consider the IID Bernoulli$(u/N)$ process, for which $\rvmu_N=u$ deterministically.
    Then $\rvLambda_N\sim\operatorname{Bin}(N,u/N)$.
    If $L^\Lambda(u)>r$, then $\{\rvLambda_N\le r\}\subseteq\{\rvLambda_N<L^\Lambda(\rvmu_N)\}$.
    Therefore
    \begin{equation}
        \begin{aligned}
            \Prob\!\left[\rvLambda_N<L^\Lambda(\rvmu_N)\right]
             &\ge
            \Prob\!\left[\rvLambda_N\le r\right]
            \\
             &=
            B_N(r,u)
            >
            \eps,
        \end{aligned}
    \end{equation}
    This contradicts the assumed lower-$\rvLambda_N$ validity.
    Hence $L^\Lambda(u)\le r$ for every such $r$.
    Applying the same argument to the complemented variables gives the analogous lower-$\rvmu_N$ and upper-$\rvLambda_N$ limits.
\end{proof}

\suppsection{Kato bounds}\label{app:kato_bounds}
For Bernoulli variables $\rvxi_1,\ldots,\rvxi_N$, Kato's concentration result~\cite{kato2020concentration} gives the following inequalities for any constants $a,b\in\mathbb R$ satisfying $b\ge |a|$.
The parameter $a$ is chosen before observing the data and can be optimized from a pre-data guess of either $\rvLambda_N$ or $\rvmu_N$~\cite{curraslorenzo2021tf}.
This section records the Kato formulas used for the numerical comparisons in the main text.

\begin{equation}
    \begin{aligned}
        &\mathbb P\!\left[
            \rvmu_N-\rvLambda_N
            \ge
            \left[
                b+a\!\left(\frac{2\rvLambda_N}{N}-1\right)
            \right]\sqrt{N}
        \right]
        \\
        &\quad \le
        \exp\!\left[
            \frac{-2(b^2-a^2)}
            {\left(1+\frac{4a}{3\sqrt N}\right)^2}
        \right].
    \end{aligned}
    \label{eq:kato_A1}
\end{equation}
Replacing $\rvxi_u$ with $1-\rvxi_u$ and $a$ with $-a$ in \eqref{eq:kato_A1} gives
\begin{equation}
    \begin{aligned}
        &\mathbb P\!\left[
            \rvLambda_N-\rvmu_N
            \ge
            \left[
                b+a\!\left(\frac{2\rvLambda_N}{N}-1\right)
            \right]\sqrt{N}
        \right]
        \\
        &\quad \le
        \exp\!\left[
            \frac{-2(b^2-a^2)}
            {\left(1-\frac{4a}{3\sqrt N}\right)^2}
        \right].
    \end{aligned}
    \label{eq:kato_A2}
\end{equation}

We now write the four inversions of Eqs.~\eqref{eq:kato_A1}--\eqref{eq:kato_A2}~\cite{curraslorenzo2021randomsampling}.
Fix a target failure probability $\varepsilon\in(0,1)$.
Define the clipping map $\Pi_N(z):=\min\{N,\max\{0,z\}\}$.
All square roots below denote the principal nonnegative branch.
Let $\rho\in\{-1,+1\}$ and define
\begin{equation}
    b_\rho(a)
    :=
    \sqrt{
        a^2
        -\frac{\ln\varepsilon}{2}
        \left(1-\rho\frac{4a}{3\sqrt N}\right)^2
    } .
    \label{eq:kato_numeric_b}
\end{equation}
With $b=b_\rho(a)$, Kato's condition $b\ge |a|$ is satisfied, and the exponential bound is exactly $\varepsilon$.
The signs $\rho=-1$ and $\rho=+1$ correspond to Eqs.~\eqref{eq:kato_A1} and~\eqref{eq:kato_A2}, respectively.

\begin{widetext}
In each paired expression below, the upper signs give the lower-bound coefficient $a_{-1}$ and the lower signs give the upper-bound coefficient $a_{+1}$.

For the $\rvmu_N$ guess $\guessmu$, the  coefficient programs are
\begin{equation}
    \begin{aligned}
        a^\Lambda_{\mp1}(\guessmu)
         &\in
        \underset{a}{\arg\max/\arg\min}
        \left\{
            \frac{N}{\sqrt N\pm 2a}
            \left[
                \frac{\guessmu}{\sqrt N}
                \pm a
                \mp b_{\mp1}(a)
            \right]
        \right\}.
    \end{aligned}
    \label{eq:kato_numeric_lambda_program}
\end{equation}

For the $\rvLambda_N$ guess $\guessLambda$, the coefficient programs are
\begin{equation}
    \begin{aligned}
        a^\mu_{\mp1}(\guessLambda)
         &\in
        \underset{a}{\arg\max/\arg\min}
        \left\{
            \guessLambda
            \mp
            \left[
                b_{\pm1}(a)
                +a\left(\frac{2\guessLambda}{N}-1\right)
            \right]\sqrt N
        \right\}.
    \end{aligned}
    \label{eq:kato_numeric_mu_program}
\end{equation}

The closed forms below are the analytical solutions of these programs~\cite{curraslorenzo2021tf,curraslorenzo2021randomsampling,navarrete2022trojan}.

For bounds on $\rvmu_N$, define
\begin{equation}
    Q_\Lambda:=9\guessLambda(N-\guessLambda)-2N\ln\varepsilon.
\end{equation}
Set
\begin{equation}
    \begin{aligned}
        a^\mu_{\mp1}(\guessLambda)
         &:=
        \min/\max\!\left\{
            \frac{
                3\!\left[
                    \mp72\sqrt N\,\guessLambda(N-\guessLambda)\ln\varepsilon
                    \pm16N^{3/2}(\ln\varepsilon)^2
                    +9\sqrt2\,(N-2\guessLambda)
                    \sqrt{-N^2\ln\varepsilon\, Q_\Lambda}
                \right]
            }
            {
                4(9N-8\ln\varepsilon)Q_\Lambda
            },
            \pm\frac{\sqrt N}{2}
        \right\}.
    \end{aligned}
    \label{eq:kato_mu_coeffs}
\end{equation}
With these coefficients, the Kato bounds on $\rvmu_N$ are
\begin{align}
    L_{\mathrm{Kato}(\guessLambda)}^{\rvmu}(s)
     &:=
    \Pi_N\!\left[
        s-
        \left[
            b_{+1}\!\left[a^\mu_{-1}(\guessLambda)\right]
            +a^\mu_{-1}(\guessLambda)\left(\frac{2s}{N}-1\right)
        \right]\sqrt N
    \right],
    \nonumber\\
    U_{\mathrm{Kato}(\guessLambda)}^{\rvmu}(s)
     &:=
    \Pi_N\!\left[
        s+
        \left[
            b_{-1}\!\left[a^\mu_{+1}(\guessLambda)\right]
            +a^\mu_{+1}(\guessLambda)\left(\frac{2s}{N}-1\right)
        \right]\sqrt N
    \right].
    \label{eq:kato_mu_bounds}
\end{align}
They satisfy
\begin{equation}
    \mathbb P\!\left[\rvmu_N\le L_{\mathrm{Kato}(\guessLambda)}^{\rvmu}(\rvLambda_N)\right]\le \varepsilon,
    \qquad
    \mathbb P\!\left[\rvmu_N\ge U_{\mathrm{Kato}(\guessLambda)}^{\rvmu}(\rvLambda_N)\right]\le \varepsilon.
    \label{eq:kato_mu_validity}
\end{equation}

For bounds on $\rvLambda_N$, set
\begin{equation}
    \begin{aligned}
        a^\Lambda_{\mp1}(\guessmu)
         &:=
        \frac{
            3\sqrt N\!\left[
                \mp9(3N^2-8N\guessmu+8\guessmu^2)\ln\varepsilon
                \mp4N(\ln\varepsilon)^2
                +9(N-2\guessmu)
                \sqrt{N\ln\varepsilon\left[N\ln\varepsilon-18\guessmu(N-\guessmu)\right]}
            \right]
        }
        {
            4\!\left[
                36(N^2-2N\guessmu+2\guessmu^2)\ln\varepsilon
                +81N\guessmu(N-\guessmu)
                +4N(\ln\varepsilon)^2
            \right]
        } .
    \end{aligned}
    \label{eq:kato_lambda_coeffs}
\end{equation}
With these coefficients, the Kato bounds on $\rvLambda_N$ are
\begin{align}
    L_{\mathrm{Kato}(\guessmu)}^{\rvLambda}(u)
     &:=
    \begin{cases}
        \Pi_N\!\left[
            \dfrac{N}{\sqrt N+2a^\Lambda_{-1}(\guessmu)}
            \left[
                \dfrac{u}{\sqrt N}
                +a^\Lambda_{-1}(\guessmu)
                -b_{-1}\!\left[a^\Lambda_{-1}(\guessmu)\right]
            \right]
        \right],
         & \sqrt N+2a^\Lambda_{-1}(\guessmu)>0,\\[3ex]
        0,
         & \text{otherwise,}
    \end{cases}
    \nonumber\\
    U_{\mathrm{Kato}(\guessmu)}^{\rvLambda}(u)
     &:=
    \begin{cases}
        \Pi_N\!\left[
            \dfrac{N}{\sqrt N-2a^\Lambda_{+1}(\guessmu)}
            \left[
                \dfrac{u}{\sqrt N}
                -a^\Lambda_{+1}(\guessmu)
                +b_{+1}\!\left[a^\Lambda_{+1}(\guessmu)\right]
            \right]
        \right],
         & \sqrt N-2a^\Lambda_{+1}(\guessmu)>0,\\[3ex]
        N,
         & \text{otherwise.}
    \end{cases}
    \label{eq:kato_lambda_bounds}
\end{align}
They satisfy
\begin{equation}
    \mathbb P\!\left[\rvLambda_N\le L_{\mathrm{Kato}(\guessmu)}^{\rvLambda}(\rvmu_N)\right]\le \varepsilon,
    \qquad
    \mathbb P\!\left[\rvLambda_N\ge U_{\mathrm{Kato}(\guessmu)}^{\rvLambda}(\rvmu_N)\right]\le \varepsilon.
    \label{eq:kato_lambda_validity}
\end{equation}
\end{widetext}

\suppsection{Azuma--Hoeffding bounds}\label{app:azuma_hoeffding_bounds}

The Azuma--Hoeffding inequality~\cite{azuma1967weighted}, applied to the martingale $\rvZ_n:=\sum_{i=1}^n(\rvxi_i-\rvp_i)=\rvLambda_n-\rvmu_n$, gives, for any $t>0$,
\begin{align}
    \mathbb P\!\left[\rvLambda_N-\rvmu_N \ge t\right]
     &\le
    e^{-\frac{2t^2}{N}},
    \nonumber \\
    \mathbb P\!\left[\rvmu_N-\rvLambda_N \ge t\right]
     &\le
    e^{-\frac{2t^2}{N}}.
\end{align}
Fix a failure probability $\varepsilon\in(0,1)$ and define
\begin{equation}
    \Delta_{\mathrm{Azuma}}(N,\varepsilon):=\sqrt{\frac{N\ln(1/\varepsilon)}{2}}.
\end{equation}
Define the bounds on $\rvmu_N$ by
\begin{align}
    U_{\mathrm{Azuma}}^{\rvmu}(s)
     &:=
    s+\Delta_{\mathrm{Azuma}}(N,\varepsilon),
    \\
    L_{\mathrm{Azuma}}^{\rvmu}(s)
     &:=
    s-\Delta_{\mathrm{Azuma}}(N,\varepsilon).
\end{align}
Then
\begin{align}
    \mathbb P\!\left[\rvmu_N \ge U_{\mathrm{Azuma}}^{\rvmu}(\rvLambda_N)\right]
     &\le \varepsilon,
    \\
    \mathbb P\!\left[\rvmu_N \le L_{\mathrm{Azuma}}^{\rvmu}(\rvLambda_N)\right]
     &\le \varepsilon.
\end{align}
Define the bounds on $\rvLambda_N$ by
\begin{align}
    U_{\mathrm{Azuma}}^{\rvLambda}(u)
     &:=
    u + \Delta_{\mathrm{Azuma}}(N,\varepsilon),
    \\
    L_{\mathrm{Azuma}}^{\rvLambda}(u)
     &:=
    u - \Delta_{\mathrm{Azuma}}(N,\varepsilon).
\end{align}
Then
\begin{align}
    \mathbb P\!\left[\rvLambda_N \ge U_{\mathrm{Azuma}}^{\rvLambda}(\rvmu_N)\right]
     &\le \varepsilon,
    \\
    \mathbb P\!\left[\rvLambda_N \le L_{\mathrm{Azuma}}^{\rvLambda}(\rvmu_N)\right]
     &\le \varepsilon.
\end{align}

\suppsection{Asymmetric passive decoy-state BB84 protocol}\label{app:passive_bb84}

Here we provide a description of the asymmetric passive decoy-state BB84 protocol from Ref.~\cite{mizutani2026passive} for completeness.

\subsection{Protocol description}

Alice and Bob agree on the protocol inputs $N$, $p_\alpha$, $p_\omega$, $\nu_\omega$, $q$, $f_{\mathrm{ec}}$, $\epsilon_{\mathrm{ph}}$, $\epsilon_{\mathrm{cor}}$, $\epsilon_{\mathrm{PA}}$, $\Delta_{\mathrm{PA}}$, and $\mathcal C$, specified in Table~\ref{tab:passive_app_inputs}.
Here $\nu_\omega$ denotes the intensity value, with $\nu_V=0$ and $0<\nu_D<\nu_S$.
The intensity probabilities obey $\sum_{\omega\in\{S,D,V\}}p_\omega=1$, and the basis probabilities obey $\sum_{\alpha\in\{Z,X\}}p_\alpha=1$.

The protocol runs as follows.
For $i=1,\ldots,N$, steps 1 and 2 below are repeated.
\begin{enumerate}
\item{}\emph{State preparation:}
Alice chooses a bit value $a_i\in\{0,1\}$ uniformly at random, a basis $\alpha_i\in\{Z,X\}$ with probabilities $p_Z$ and $p_X$, and an intensity label $\omega_i\in\{S,D,V\}$ with probabilities $p_S,p_D \text{ and }p_V$.
She sends Bob a phase-randomized weak coherent pulse with intensity $\nu_{\omega_i}$ and BB84 polarization determined by $(a_i,\alpha_i)$.

\item{}\emph{Asymmetric passive measurement:}
Bob sends the incoming pulse through a beam splitter of ratio $q\in(0,1/2)$ and uses four threshold single-photon detectors to measure the $Z$ and $X$ arms.
If no detector clicks, Bob records a no-click event.
If clicks occur only on the $Z$ ($X$) arm, Bob sets $\beta_i=Z$ ($X$) and assigns the bit value $b_i$ from the clicked detector.
If both detectors in one arm click and the other arm is silent, Bob assigns $b_i$ uniformly at random.
If detectors in both arms click, Bob records a cross-click event.
\end{enumerate}

The post-processing and the public discussion over an authenticated channel run as follows:
\begin{enumerate}
\setcounter{enumi}{2}
\item{}\emph{Data announcement:}
Bob records the index sets $\mathcal D_Z$, $\mathcal D_X$ and $\mathcal D_{\mathrm{cross}}$ for events with $\beta_i=Z$, events with $\beta_i=X$, and cross-click events.
He obtains the $X$-basis string $\bm{b}_X:=(b_i)_{i\in\mathcal D_X}$.
He announces $\mathcal D_Z$, $\mathcal D_X$, $\mathcal D_{\mathrm{cross}}$, and $\bm{b}_X$ over an authenticated classical channel.

\item{}\emph{Sifting:}
Alice forms
\begin{equation}
    \begin{aligned}
        \mathcal D'_Z
         &:=
        \{i\in\mathcal D_Z:\alpha_i=Z\},
        \\
        N_{\mathrm{sift}}
         &:=
        |\mathcal D'_Z|.
    \end{aligned}
\end{equation}
Alice announces $(\alpha_i)_{i\in\mathcal D'_Z}$ to Bob, who obtains his sifted key $\bm{b}_B:=(b_i)_{i\in\mathcal D'_Z}$.
Alice obtains her sifted key $\bm{a}_A:=(a_i)_{i\in\mathcal D'_Z}$.

\item{}\emph{Parameter estimation:}
Alice forms the observable counters
\begin{equation}
    \begin{aligned}
        N_{Z,\omega}
         &:=
        |\{i\in\mathcal D'_Z:\omega_i=\omega\}|,
        \\
        N_{X,\omega}^{\mathrm{Error}}
         &:=
        |\{i\in\mathcal D_X:\alpha_i=X,
        \\
        &\qquad
        \omega_i=\omega,a_i\ne b_i\}|,
        \\
        N_{\omega}^{\mathrm{Cross}}
         &:=
        |\{i\in\mathcal D_{\mathrm{cross}}:\omega_i=\omega\}|.
    \end{aligned}
\end{equation}
The parties use the counters $\{N_{Z,\omega}\}_{\omega}$, $\{N_{X,\omega}^{\mathrm{Error}}\}_{\omega}$, and $\{N_{\omega}^{\mathrm{Cross}}\}_{\omega}$ to evaluate a lower bound $N_Z^{1,L}$ on the single-photon $Z$-basis counts and an upper bound $N_{\mathrm{ph}}^{1,U}$ on the corresponding phase-error counts.
In Ref.~\cite{mizutani2026passive}, these bounds are obtained with Kato's inequality~\cite{kato2020concentration}.
Here we obtain them from the concentration family $\mathcal C$ described in the concentration-families subsection below.
The privacy-amplification amount is set to
\begin{equation}
    N_{\mathrm{PA}}
    =
    N_{\mathrm{sift}}
    -
    N_Z^{1,L}
    \left[
        1-h\!\left(\frac{N_{\mathrm{ph}}^{1,U}}{N_Z^{1,L}}\right)
    \right]
    +\Delta_{\mathrm{PA}},
    \label{eq:passive_app_npa}
\end{equation}
where $h(x):=-x\log_2 x-(1-x)\log_2(1-x)$ is the binary entropy function.

\item{}\emph{Error correction:}
Alice announces a bit-error-correcting code and sends a syndrome on $\bm{a}_A$ encrypted via the one-time pad by consuming a pre-shared secret key of length $N_{\mathrm{EC}}$.
Bob uses the syndrome to correct $\bm{b}_B$ and obtains the reconciled key.

\item{}\emph{Privacy amplification:}
Alice and Bob perform privacy amplification on their reconciled keys, sacrificing $N_{\mathrm{PA}}$ bits as determined above, to obtain final keys $\bm{k}_A$ and $\bm{k}_B$.
Their length is given by
\begin{equation}
    \ell
    =
    \max\!\left\{
    0,\,
    N_{\mathrm{sift}}
    -N_{\mathrm{PA}}
    -N_{\mathrm{EC}}
    \right\}.
    \label{eq:passive_app_key_length}
\end{equation}
\end{enumerate}

\begin{center}
    \refstepcounter{table}
    \label{tab:passive_app_inputs}
    \small
    \textbf{TABLE~\thetable.}
    Asymmetric passive decoy-state BB84 inputs
    \begin{tabular}{@{}ll@{}}
        \hline\hline
        Symbol & Meaning \\
        \hline
        $N$ & Number of emitted pulses \\
        $p_\alpha$ & Alice basis probability, $\alpha\in\{Z,X\}$ \\
        $p_\omega$ & Intensity-label probability, $\omega\in\{S,D,V\}$ \\
        $\nu_\omega$ & Intensity for label $\omega$\\
        $q$ & Passive beam-splitter ratio \\
        $f_{\mathrm{ec}}$ & Error-correction leakage efficiency \\
        $\eps_{\mathrm{ph}}$ & Failure probability of the parameter-estimation step \\
        $\eps_{\mathrm{cor}}$ & Failure probability of the error-correction step \\
        $\eps_{\mathrm{PA}}$ & Failure probability of the privacy-amplification step \\
        $\Delta_{\mathrm{PA}}$ & Privacy-amplification overhead, with $\eps_{\mathrm{PA}}=2^{-\Delta_{\mathrm{PA}}}$ \\
        $\mathcal C$ & Concentration family for finite-statistics bounds \\
        \hline\hline
    \end{tabular}
\end{center}

\subsection{Simulation details}

Let $\eta_{\mathrm{sys}}$ denote the overall system transmittance including detector efficiency (which, for simplicity, we assume is the same for all detectors). Also, let $d$ be the per-detector dark-count probability, and let $\delta_{\mathrm{mis}}$ be the misalignment contribution from the channel and Bob's receiver. 
For the satellite-loss plots, a loss value $L$ in dB is converted to the transmittance
\begin{equation}
    \eta_{\mathrm{sys}}=10^{-L/10}.
\end{equation}
For a pulse of intensity $\nu$, we define the single-round detection and error probabilities (neglecting for the moment the effect of $\delta_{\mathrm{mis}}$) as follows
\begin{align}
    P_Z(\nu)
     & =
    e^{-\nu\eta_{\mathrm{sys}}q}(1-d)^2
    \nonumber \\
     &\quad \times
    \Bigl[
        1-e^{-\nu\eta_{\mathrm{sys}}(1-q)}(1-d)^2
    \Bigr],
    \\
    P_X(\nu)
     & =
    e^{-\nu\eta_{\mathrm{sys}}(1-q)}(1-d)^2
    \nonumber \\
     &\quad \times
    \Bigl[
        1-e^{-\nu\eta_{\mathrm{sys}}q}(1-d)^2
    \Bigr],
    \\
    P_Z^{\mathrm{err}}(\nu)
     & =
    e^{-\nu\eta_{\mathrm{sys}}q}(1-d)^2
    \nonumber \\
     &\quad \times
    \Bigl\{
        e^{-\nu\eta_{\mathrm{sys}}(1-q)}
        \bigl[d(1-d)+d^2/2\bigr]
    \nonumber \\
     &\qquad
        +\bigl[1-e^{-\nu\eta_{\mathrm{sys}}(1-q)}\bigr]d/2
    \Bigr\},
    \\
    P_X^{\mathrm{err}}(\nu)
     & =
    e^{-\nu\eta_{\mathrm{sys}}(1-q)}(1-d)^2
    \nonumber \\
     &\quad \times
    \Bigl\{
        e^{-\nu\eta_{\mathrm{sys}}q}
        \bigl[d(1-d)+d^2/2\bigr]
    \nonumber \\
     &\qquad
        +\bigl[1-e^{-\nu\eta_{\mathrm{sys}}q}\bigr]d/2
    \Bigr\}.
\end{align}
The cross-click probability is
\begin{align}
    P_{\mathrm{Cross}}(\nu)
     & =
    \bigl[1-e^{-\nu\eta_{\mathrm{sys}}(1-q)}\bigr]
    \bigl[1-e^{-\nu\eta_{\mathrm{sys}}q}\bigr]
    \nonumber \\
     & \quad
    +\bigl[1-e^{-\nu\eta_{\mathrm{sys}}(1-q)}\bigr]
    e^{-\nu\eta_{\mathrm{sys}}q}
    \bigl[1-(1-d)^2\bigr]
    \nonumber \\
     & \quad
    +\bigl[1-e^{-\nu\eta_{\mathrm{sys}}q}\bigr]
    e^{-\nu\eta_{\mathrm{sys}}(1-q)}
    \bigl[1-(1-d)^2\bigr]
    \nonumber \\
     & \quad
    +e^{-\nu\eta_{\mathrm{sys}}}
    \bigl[1-(1-d)^2\bigr]^2 .
\end{align}
For each $\mathcal B\in\{Z,X\}$ and $\omega\in\{S,D,V\}$, the simulated number of counts are given by
\begin{equation}
    N_{\mathcal B,\omega}=N p_{\mathcal B}p_\omega P_{\mathcal B}(\nu_\omega),
    \label{eq:passive_app_nb}
\end{equation}
and
\begin{equation}
    N_{\mathrm{sift}}=\sum_{\omega\in\{S,D,V\}} N_{Z,\omega}.
\end{equation}
The corresponding $Z$-basis bit-error rate (now including the effect of $\delta_{\mathrm{mis}}$) is modeled as
\begin{equation}
    e_{\mathrm{bit}}
    =
    \frac{N_{\mathrm{sift}}^{\mathrm{Error}}}{N_{\mathrm{sift}}}
    +\delta_{\mathrm{mis}},
\end{equation}
where
\begin{equation}
    N_{\mathrm{sift}}^{\mathrm{Error}}
    =
    \sum_{\omega\in\{S,D,V\}}
    Np_Zp_\omega P_Z^{\mathrm{err}}(\nu_\omega).
\end{equation}
The $X$-basis error counts are
\begin{equation}
    \begin{aligned}
        N_{X,\omega}^{\mathrm{Error}}
         & =
        Np_Xp_\omega P_X^{\mathrm{err}}(\nu_\omega)
        +\delta_{\mathrm{mis}} N_{X,\omega},
    \end{aligned}
    \label{eq:passive_app_nx_error}
\end{equation}
where the cross-click counts are
\begin{equation}
    N_{\omega}^{\mathrm{Cross}}
    =
    Np_\omega P_{\mathrm{Cross}}(\nu_\omega).
    \label{eq:passive_app_cross}
\end{equation}
The error-correction leakage is modelled as
\begin{equation}
    N_{\mathrm{EC}}=f_{\mathrm{ec}} N_{\mathrm{sift}} h(e_{\mathrm{bit}}).
\end{equation}
For all satellite key-rate evaluations, the error-correction leakage efficiency is assumed to be $f_{\mathrm{ec}}=1.16$, and the privacy-amplification overhead is $\Delta_{\mathrm{PA}}=71$.
We set $\eps_{\mathrm{cor}}=5\times10^{-11}$ and assign each individual concentration call the failure probability $\eps=10^{-20}/144$, so $\eps_{\mathrm{ph}}=9\eps$ and $\eps_{\mathrm{sec}}\le\eps_{\mathrm{cor}}+\sqrt{2(\eps_{\mathrm{ph}}+\eps_{\mathrm{PA}})}\approx10^{-10}$.

\subsection{Concentration-families}\label{app:passive_bb84_concentration_families}
Let $\mathcal C$ denote any concentration family that supplies the four bounds
$
    U_{\mathcal C}^{\mu}(s),
    L_{\mathcal C}^{\mu}(s),
    U_{\mathcal C}^{\Lambda}(u) \text{ and }
    L_{\mathcal C}^{\Lambda}(u).
    $
Define the photon-number distribution and the averaged one-photon probability by
\begin{align}
    p_\omega^{\mathrm{int}}(n)
     &:=
    e^{-\nu_\omega}\nu_\omega^n/n!,
    \\
    p^{\mathrm{int}}(1)
     &:=
    \sum_{\omega\in\{S,D,V\}} p_\omega p_\omega^{\mathrm{int}}(1),
\end{align}
and let the parameters $c_1 \text{and } c_2 $ be given by
\begin{equation}
    c_1=\frac{(1-q)p_Z}{q p_X},
    \qquad
    c_2=
    \frac{p_Z(1-q)(1-2q)(1-d)^2}
    {1-\bigl[q^2+(1-q)^2\bigr](1-d)^2}.
\end{equation}
For the phase-error bound, define the intermediate quantity~\cite{mizutani2026passive}
\begin{align}
    \mu_{\mathrm{ph},\mathcal C}^{U}
     & =
    \frac{p^{\mathrm{int}}(1)}{\nu_D}
    \Biggl[
        c_1\biggl(
        e^{\nu_D}\frac{U_{\mathcal C}^{\mu}(N_{X,D}^{\mathrm{Error}})}{p_D}
        -\frac{L_{\mathcal C}^{\mu}(N_{X,V}^{\mathrm{Error}})}{p_V}
        \biggr)
    \nonumber \\
     &\quad
        +c_2\biggl(
        e^{\nu_D}\frac{U_{\mathcal C}^{\mu}(N_{D}^{\mathrm{Cross}})}{p_D}
        -\frac{L_{\mathcal C}^{\mu}(N_{V}^{\mathrm{Cross}})}{p_V}
        \biggr)
        \Biggr].
    \label{eq:passive_app_mu_ph}
\end{align}
The resulting phase-error bound is
\begin{equation}
    \begin{aligned}
    N_{\mathrm{ph},\mathcal C}^{1,U}
     & =
    U_{\mathcal C}^{\Lambda}\!\left(\mu_{\mathrm{ph},\mathcal C}^{U}\right).
    \end{aligned}
    \label{eq:passive_app_nph}
\end{equation}
For the single-photon $Z$-basis decoy reduction, define~\cite{mizutani2026passive}
\begin{align}
    \mu_{1Z,\mathcal C}^{L}
     & =
    p^{\mathrm{int}}(1)
    \Biggl[
        -\frac{\nu_D e^{\nu_S}}{p_S\nu_S(\nu_S-\nu_D)}
        U_{\mathcal C}^{\mu}(N_{Z,S})
    \nonumber \\
     &\quad
        +\frac{\nu_S e^{\nu_D}}{p_D\nu_D(\nu_S-\nu_D)}
        L_{\mathcal C}^{\mu}(N_{Z,D})
    \nonumber \\
     &\quad
        -\frac{\nu_S}{p_V\nu_D(\nu_S-\nu_D)}
        U_{\mathcal C}^{\mu}(N_{Z,V})
        \Biggr],
    \label{eq:passive_app_mu_1z}
\end{align}
The resulting single-photon $Z$-basis bound is
\begin{equation}
    \begin{aligned}
    N_{Z,\mathcal C}^{1,L}
     & =
    L_{\mathcal C}^{\Lambda}\!\left(\mu_{1Z,\mathcal C}^{L}\right).
    \end{aligned}
    \label{eq:passive_app_n1z}
\end{equation}

\suppsection{Additional concentration and satellite key-rate plots}\label{app:additional_satellite_keyrate_plots}

\subsection{Full-range concentration comparison}
\label{subsec:app_full_range_concentration}

Figure~\ref{fig:app-mu-gap-full} compares the Kato and fixed-$\lambda$ upper-$\rvmu_N$ bounds against the exact binomial benchmark.
The inset shows that near the calibration point $g_\Lambda/N=0.1$, the fixed-$\lambda$ curve overlaps Kato's bound, with both lying above the binomial benchmark.
Across the full range, Kato's gap grows toward large $x$, whereas the fixed-$\lambda$ curve follows the binomial benchmark's downturn as $x$ approaches one.

\begin{figure}[!htbp]
    \centering
    \includegraphics[width=0.98\columnwidth]{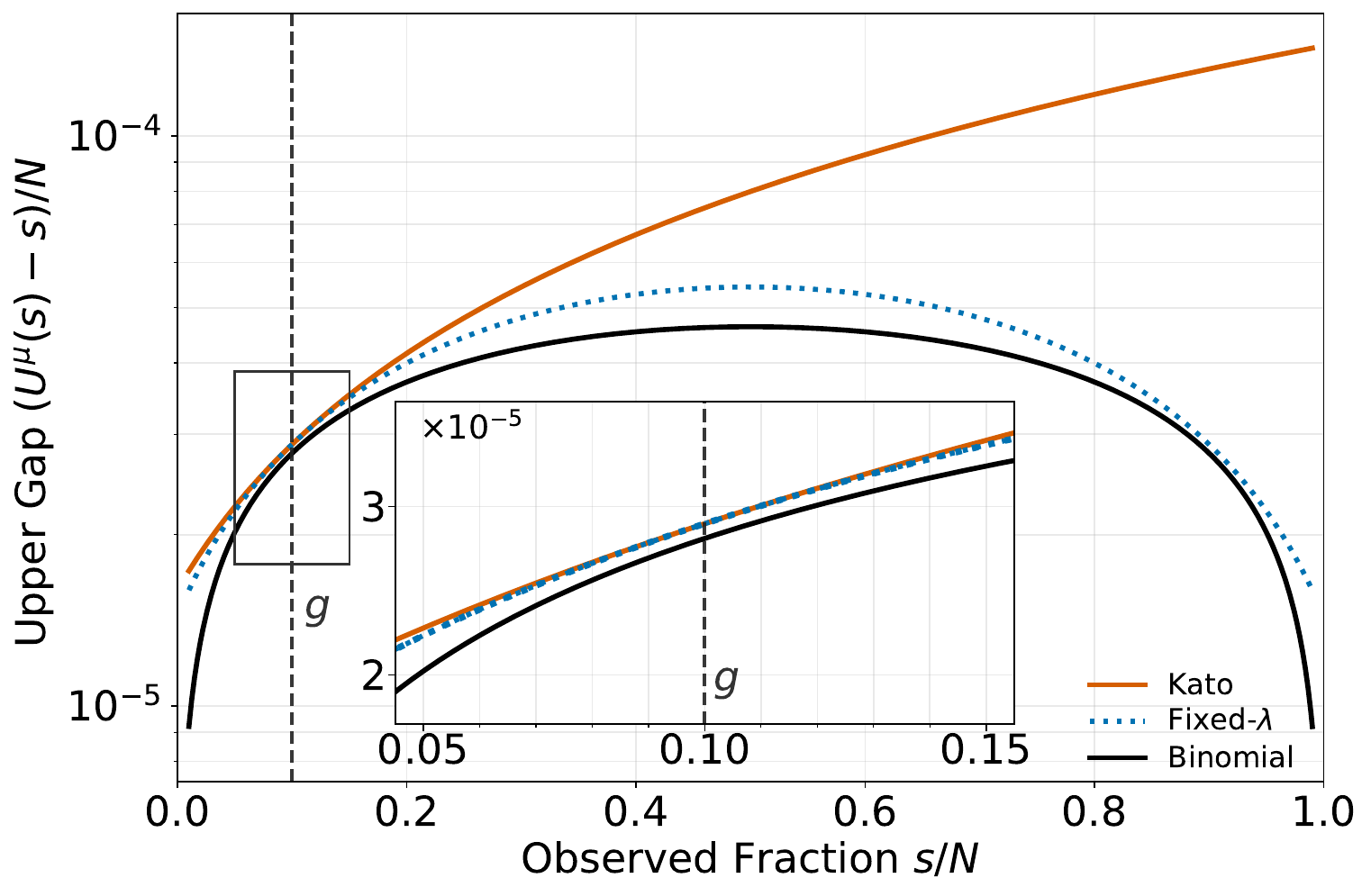}
    \caption{Full-range upper-$\mu_N$ concentration benchmarks for $N=10^{10}$ and $\eps=10^{-20}$.
    The horizontal axis displays $0\le x=s/N\le1$.
    The inset zooms the region $0.05\le x\le0.15$ around the normalized guess $g_\Lambda/N=0.1$.
    The dashed vertical line marks this normalized reference guess.
    The blue dotted curve is the fixed-$\lambda$ bound of Proposition~1 of the Article, calibrated at this guess.
    The remaining curves are Kato's bound~\cite{kato2020concentration} (orange) and the exact Clopper--Pearson binomial inversion (black).}
    \label{fig:app-mu-gap-full}
\end{figure}

\subsection{Union bounds}
\label{subsec:app_union_bounds}

Let $\mathcal M$ denote either the numerical or analytical mixture family and $\mathcal F$ denote the fixed-$\lambda$ family.
For a target one-sided failure probability $\eps$, all component maps on the right-hand side below are evaluated at failure probability $\eps/2$.
Define the union family $\mathcal U$ by
\begingroup
\small
\begin{align}
    U_{\mathcal U}^{\mu}(s)
     &:=
    \min\!\left\{
        U_{\mathcal M}^{\mu}(s),
        U_{\mathcal F}^{\mu}(s)
    \right\},
    \nonumber\\
    L_{\mathcal U}^{\mu}(s)
     &:=
    \max\!\left\{
        L_{\mathcal M}^{\mu}(s),
        L_{\mathcal F}^{\mu}(s)
    \right\},
    \nonumber\\
        U_{\mathcal U}^{\Lambda}(u)
     &:=
    \min\!\left\{
        U_{\mathcal M}^{\Lambda}(u),
        U_{\mathcal F}^{\Lambda}(u)
    \right\},
    \label{eq:app_union_bounds}
    \nonumber\\
    L_{\mathcal U}^{\Lambda}(u)
     &:=
    \max\!\left\{
        L_{\mathcal M}^{\Lambda}(u),
        L_{\mathcal F}^{\Lambda}(u)
    \right\}.
\end{align}
\endgroup
For the upper-$\rvmu_N$ bound, set
\begin{equation}
    E_{\mathcal C}^{\mu}
    :=
    \left\{\rvmu_N>U_{\mathcal C}^{\mu}(\rvLambda_N)\right\},
    \qquad
    \mathcal C\in\{\mathcal M,\mathcal F\}.
    \label{eq:app_union_mu_event}
\end{equation}
Then
\begin{align}
    \Prob\!\left[\rvmu_N>U_{\mathcal U}^{\mu}(\rvLambda_N)\right]
     &=
    \Prob\!\left[E_{\mathcal M}^{\mu}\cup E_{\mathcal F}^{\mu}\right]
    \nonumber\\
     &\le
    \Prob\!\left[E_{\mathcal M}^{\mu}\right]
    +
    \Prob\!\left[E_{\mathcal F}^{\mu}\right]
    \nonumber\\
     &\le
    \frac{\eps}{2}+\frac{\eps}{2}
    =
    \eps.
    \label{eq:app_union_bound_example}
\end{align}
The same argument gives the corresponding guarantees for $L_{\mathcal U}^{\mu}$, $L_{\mathcal U}^{\Lambda}$ and $U_{\mathcal U}^{\Lambda}$.

\subsection{Expected-loss optimization}
\label{subsec:app_expected_loss_optimization}

Figures~\ref{fig:app-leo-expected-key-seven} and~\ref{fig:app-geo-expected-key-seven} are the main-text expected-loss-optimized key-rate plots with the union bounds of Eq.~\eqref{eq:app_union_bounds} added.
Thus the protocol parameters, expected losses, and observed-loss panels are the same as in Figs.~2 and~3 of the Article.
We find that the union-bound curves remain above the Kato curve and lower the positive-key emitted-pulse threshold.
Comparing the Kato curve with the union-bound curve, the latter lowers this threshold by about $9\%$ and $66\%$ for the LEO panels at $L_{\rm obs}=35$ and $40\,\mathrm{dB}$, respectively.
For the GEO panels, the corresponding reductions are about $13\%$ and $36\%$ at $L_{\rm obs}=55$ and $60\,\mathrm{dB}$, respectively.

\begin{figure}[!htbp]
    \centering
    \columnratepanel{(a)}{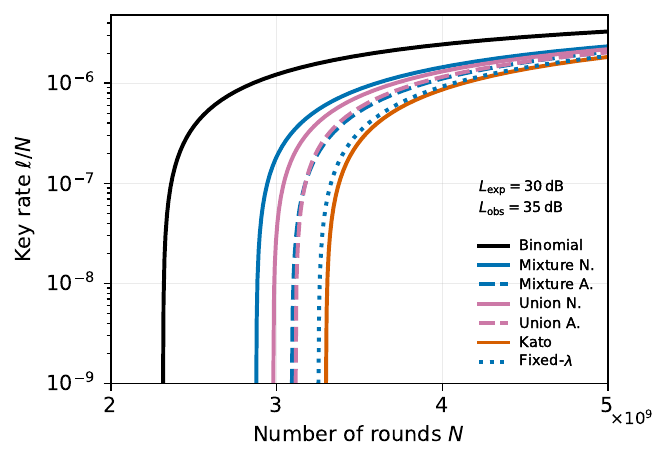}\\[0.6em]
    \columnratepanel{(b)}{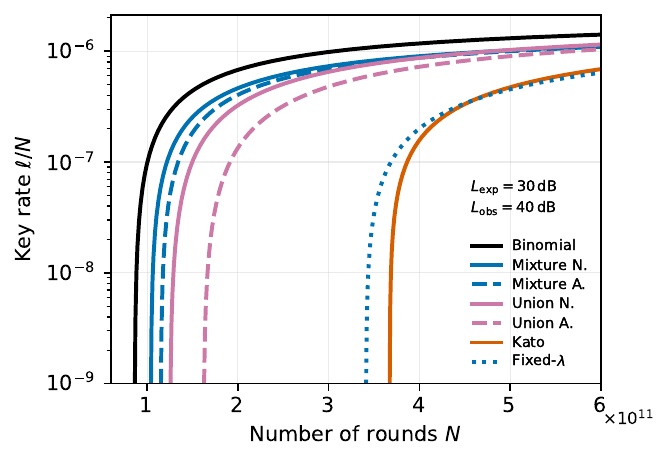}
    \caption{Finite secret key rate versus number of rounds for an asymmetric passive decoy-state BB84 protocol over a LEO-satellite link with channel parameters $d=10^{-7}$ and $\delta_{\mathrm{mis}}=0.03$.
    The expected loss is $L_{\rm exp}=30\,\mathrm{dB}$, while panels (a) and (b) use $L_{\rm obs}=35$ and $40\,\mathrm{dB}$, respectively.
    Protocol parameters are optimized at the expected loss separately for each concentration family.
    Blue curves use the numerical mixture-martingale bounds of Proposition~1 of the Article with $[A,B]=[0,10]$ (solid), their corresponding analytical relaxations (dashed), and the fixed-$\lambda$ martingale bound of Proposition~1 of the Article (dotted).
    Pink curves use the union bounds defined in Eq.~\eqref{eq:app_union_bounds} (numerical union, solid; analytical union, dashed).
    The remaining curves use the Clopper--Pearson binomial benchmark (black) and Kato's bound~\cite{kato2020concentration} (orange).}
    \label{fig:app-leo-expected-key-seven}
\end{figure}

\begin{figure}[!htbp]
    \centering
    \columnratepanel{(a)}{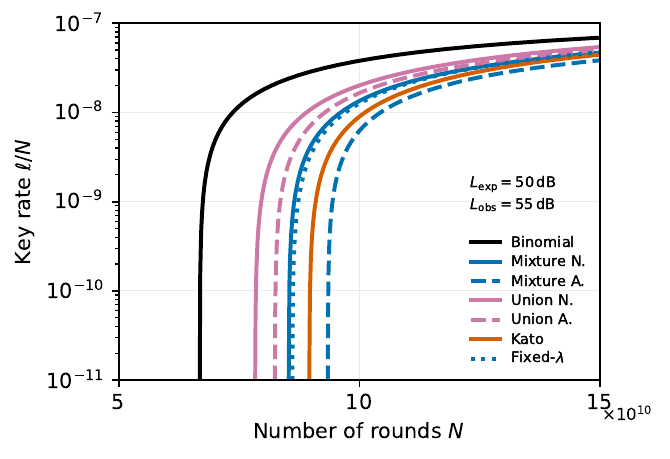}\\[0.6em]
    \columnratepanel{(b)}{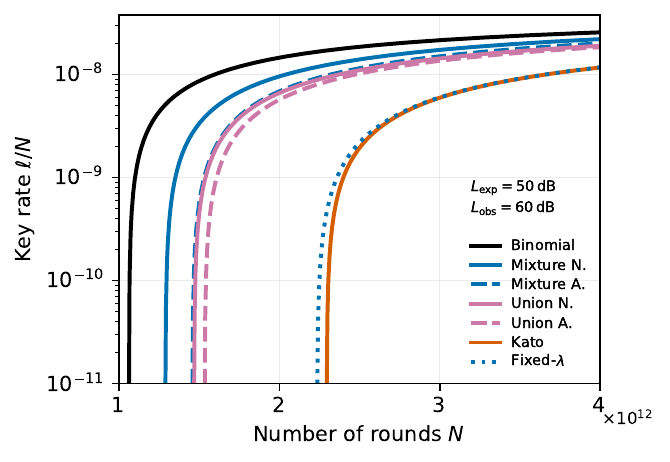}
    \caption{Finite secret key rate versus number of rounds for an asymmetric passive decoy-state BB84 protocol over a GEO-satellite link with channel parameters $d=10^{-9}$ and $\delta_{\mathrm{mis}}=0.01$.
    The expected loss is $L_{\rm exp}=50\,\mathrm{dB}$, while panels (a) and (b) use $L_{\rm obs}=55$ and $60\,\mathrm{dB}$, respectively.
    Protocol parameters are optimized at the expected loss separately for each concentration family.
    All color and line-style conventions match Fig.~\ref{fig:app-leo-expected-key-seven}.}
    \label{fig:app-geo-expected-key-seven}
\end{figure}

\begin{figure*}[!t]
    \centering
    \ratepanel{(a)}{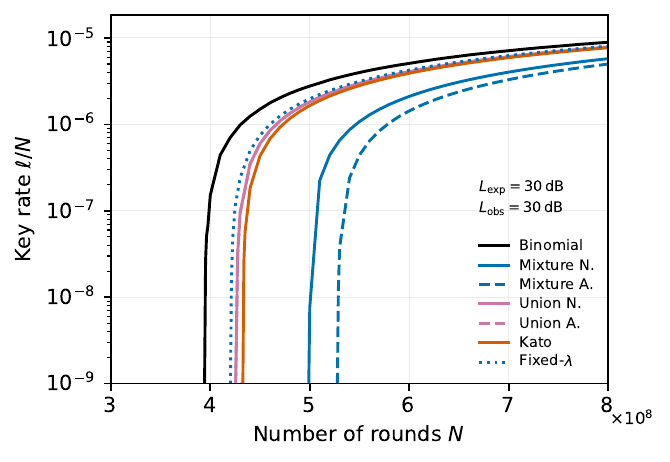}\hfill
    \ratepanel{(b)}{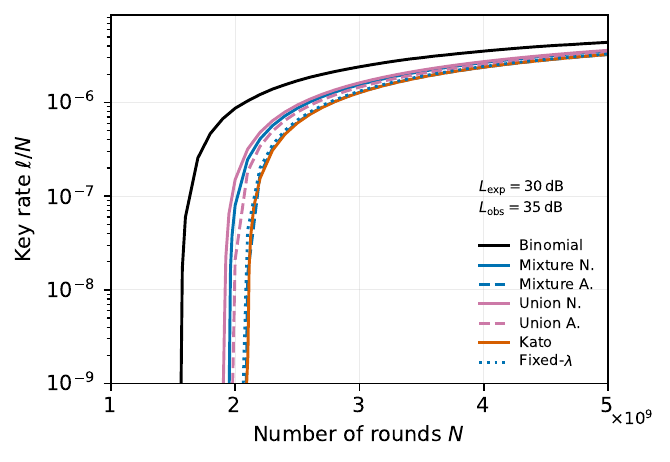}\\[0.8em]
    \ratepanel{(c)}{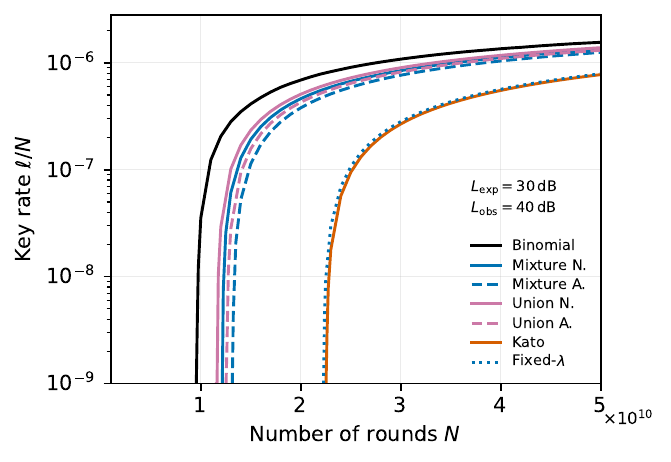}\hfill
    \ratepanel{(d)}{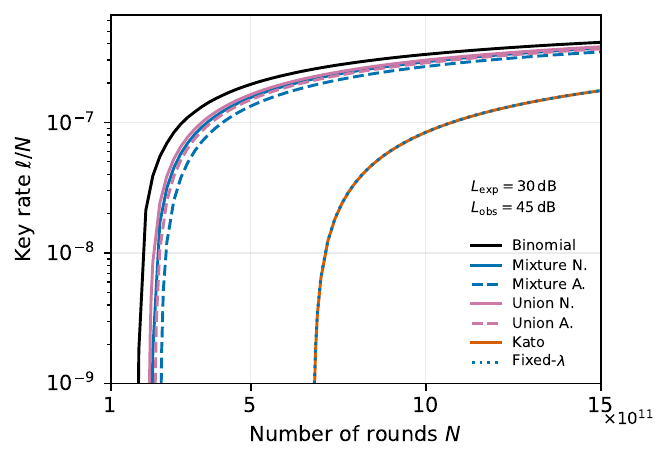}
    \caption{Finite secret key rate versus number of rounds for an asymmetric passive decoy-state BB84 protocol over a LEO-satellite link with channel parameters $d=10^{-7}$ and $\delta_{\mathrm{mis}}=0.03$.
    The expected loss is $L_{\rm exp}=30\,\mathrm{dB}$.
    Panels (a)--(d) use observed losses $L_{\rm obs}=30$, $35$, $40$, and $45\,\mathrm{dB}$, respectively.
    Each displayed $L_{\rm obs}$ is treated as a fixed design loss, and protocol parameters are optimized separately for each concentration family.
    In panel (a), the numerical and analytical union curves are visually indistinguishable at the plotted scale.
    All color and line-style conventions match Fig.~\ref{fig:app-leo-expected-key-seven}.}
    \label{fig:app-leo-observed-key}
\end{figure*}

\subsection{Observed-loss optimization}
\label{subsec:app_observed_loss_optimization}

Figures~\ref{fig:app-leo-observed-key} and~\ref{fig:app-geo-observed-key} instead optimize the protocol parameters at the observed loss $L_{\rm obs}$.
The guess values for the concentration bounds are still calculated under the expected channel $L_{\rm exp}$ for these observed-optimal intensities and probabilities.
They extend the mismatch comparison by adding the no-mismatch panels $L_{\rm obs}=L_{\rm exp}$ and the $15\,\mathrm{dB}$ mismatch panels to the $5$ and $10\,\mathrm{dB}$ cases shown in the main text.
In the no-mismatch panels, expected-loss optimization and observed-loss optimization coincide.
The Kato and fixed-$\lambda$ references are still constructed from the expected-loss design model, so the mismatched panels isolate concentration-bound mismatch from protocol-parameter mismatch.
Across all displayed observed-loss-optimized panels, the union-bound curves remain above the Kato curve.
Comparing the Kato curve with the union-bound curve in each panel, the latter lowers the positive-key emitted-pulse threshold by about $1\%$, $9\%$, $49\%$, and $69\%$ in the LEO panels and about $5\%$, $11\%$, $42\%$, and $68\%$ in the GEO panels.

\begin{figure*}[!htbp]
    \centering
    \ratepanel{(a)}{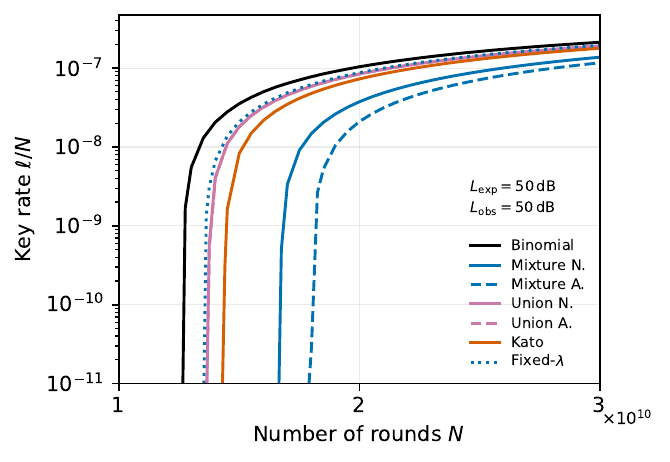}\hfill
    \ratepanel{(b)}{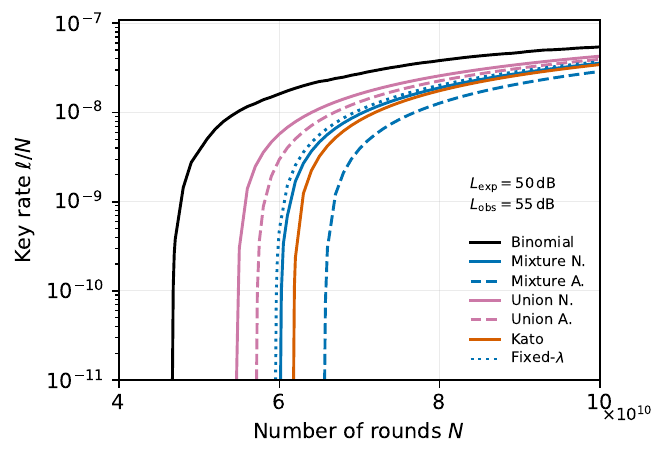}\\[0.8em]
    \ratepanel{(c)}{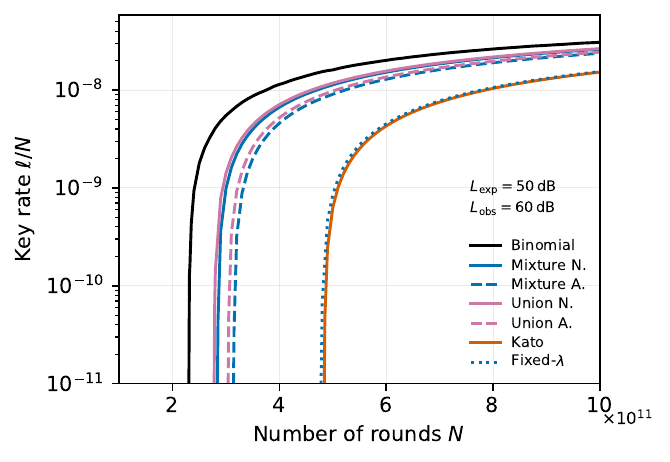}\hfill
    \ratepanel{(d)}{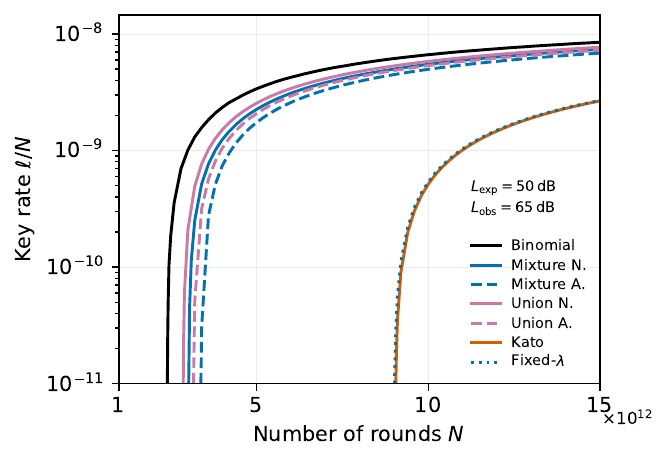}
    \caption{Finite secret key rate versus number of rounds for an asymmetric passive decoy-state BB84 protocol over a GEO-satellite link with channel parameters $d=10^{-9}$ and $\delta_{\mathrm{mis}}=0.01$.
    The expected loss is $L_{\rm exp}=50\,\mathrm{dB}$.
    Panels (a)--(d) use observed losses $L_{\rm obs}=50$, $55$, $60$, and $65\,\mathrm{dB}$, respectively.
    Each displayed $L_{\rm obs}$ is treated as a fixed design loss, and protocol parameters are optimized separately for each concentration family.
    In panel (a), the numerical and analytical union curves are visually indistinguishable at the plotted scale.
    All color and line-style conventions match Fig.~\ref{fig:app-leo-expected-key-seven}.}
    \label{fig:app-geo-observed-key}
\end{figure*}

\clearpage
\bibliography{refs}

\let\bibliography\combinedbibliography
\clearpage
\bibliography{refs}

\end{document}